\documentclass[aps,prd,10pt,notitlepage,nofootinbib,superscriptaddress,showkeys,showpacs]{revtex4-2}

\usepackage{amsmath,amssymb,amsthm,latexsym,bbm}
\usepackage[english]{babel}
\usepackage{color}
\usepackage{xspace}
\usepackage{graphicx}
\usepackage{pifont,dsfont}
\usepackage{marvosym}
\usepackage{slashed}
\usepackage{enumitem}
\usepackage{multirow}
\usepackage{hyperref}

\include{epsf}
\allowdisplaybreaks

\usepackage{mathtools}
\usepackage{graphicx}
\usepackage{stackrel}
\usepackage{color}
\usepackage{stmaryrd}
\usepackage{amsfonts}
\usepackage{amscd}
\usepackage[all]{xy}

\newcommand{\bea}{\begin{eqnarray}}
\newcommand{\eea}{\end{eqnarray}}

\newcommand{\oT}{\overline{T}}

\theoremstyle{definition}
\newtheorem{lemma}{Lemma}

\newtheorem{theorem}{Theorem}

\newtheorem{proposition}{Proposition}

\newcommand{\ba}{\begin{equation}\begin{aligned}}
\newcommand{\ea}{\end{aligned}\end{equation}}

\newtheorem{proof*}{Proof}

\newcommand{\MM}{\text{MM}}
\newcommand{\Tensor}{\text{Tensor}}

\DeclareMathOperator{\tr}{tr}

\newcommand{\vell}{\boldsymbol{\ell}}
\newcommand{\vlam}{\boldsymbol{\lambda}}
\newcommand{\vmu}{\boldsymbol{\mu}}
\SetSymbolFont{stmry}{bold}{U}{stmry}{m}{n}

\newcommand{\C}{\mathbbm{C}}

\makeatletter
\newcommand{\subalign}[1]{%
  \vcenter{%
    \Let@ \restore@math@cr \default@tag
    \baselineskip\fontdimen10 \scriptfont\tw@
    \advance\baselineskip\fontdimen12 \scriptfont\tw@
    \lineskip\thr@@\fontdimen8 \scriptfont\thr@@
    \lineskiplimit\lineskip
    \ialign{\hfil$\m@th\scriptstyle##$&$\m@th\scriptstyle{}##$\hfil\crcr
      #1\crcr
    }%
  }%
}
\makeatother

\begin{document}

\title{\Large \bf Blobbed topological recursion for correlation functions in tensor models}

\author{{\bf Valentin Bonzom}}\email{bonzom@lipn.univ-paris13.fr}
\author{{\bf Nicolas Dub}}\email{dub@lipn.univ-paris13.fr}
\affiliation{Universit\'e Sorbonne Paris Nord, LIPN, CNRS, UMR 7030, F-93430 Villetaneuse, France, EU}

\date{\small\today}

\begin{abstract}
\noindent Tensor models are generalizations of matrix models and as such, it is a natural question to ask whether they satisfy some form of the topological recursion. The world of unitary-invariant observables is however much richer in tensor models than in matrix models. It is therefore {\it a priori} unclear which set of observables could satisfy the topological recursion. Such a set of observables was identified a few years ago in the context of the quartic melonic model by the first author and Dartois. It was shown to satisfy an extension of the topological recursion introduced by Borot and called the blobbed topological recursion. Here we show that this set of observables is present in arbitrary tensor models which have non-vanishing couplings for the quartic melonic interactions. It satisfies the blobbed topological recursion in a universal way, i.e. independently of the choices of the other interactions. In combinatorial terms, the correlation functions describe stuffed maps with colored boundary components. The specifics of the model only appear in the generating functions of the stuffings and the blobbed topological recursion only requires them to have well-defined $1/N$ expansions. The spectral curve is a disjoint union of Gaussian spectral curves, with the cylinder function receiving an additional holomorphic part. This result is achieved via a perturbative rewriting of tensor models as multi-matrix models due to the first author, Lionni and Rivasseau. It is then possible to formally integrate all degrees of freedom except those which enter the topological recursion, meaning interpreting the Feynman graphs as stuffed maps. We further provide new expressions to relate the expectations of $U(N)^d$-invariant observables on the tensor and matrix sides.
\end{abstract}

\medskip


\maketitle


\section*{Introduction}

\paragraph{Matrix and tensor models --} 
Tensor models are a generalization of matrix models where the variables are the elements $T_{a_1\dotsb a_d}$ of a tensor, for $a_1, \dotsc, a_d = 1, \dotsc, N$, and their complex conjugate elements $\overline{T}_{a_1 \dotsb a_d}$. It is well-known that matrix models are intimately connected to combinatorial maps: the latter are generated by the Feynman expansion of the former \cite{MatrixModels, MatrixModelCombinatorics}. For instance, in the Hermitian matrix model,
\begin{equation}
\int dM\ e^{-\frac{N}{2t} \tr M^2 + N \sum_{k\geq 1} \frac{p_k}{k} \tr M^k} = \sum_{\text{Maps}} \frac{t^n}{n!} N^{2-2g} \prod_{k\geq 1} p_k^{n_k}
\end{equation}
where the sum is over maps with $n$ labeled edges, of genus $g$ and with $n_k$ faces of degree $k\geq 1$. The quantity in the exponential on the left hand side is called the action, or the potential, and the $p_k$s are called the coupling constants. In tensor models, this relationship is also generalized, meaning that the Feynman expansion of tensor models generates piecewise-linear $d$-dimensional (pseudo-)manifolds \cite{GurauBook, InvitationGurau, UniversalityClasses, 3D}. This is why tensor models were already proposed as candidate for quantum gravity in the early 90s, long before a large $N$ limit was found \cite{Gurau1/N}.

More recently, tensor models have been shown to provide the same large $N$ limit as the SYK model (a model which is exactly solvable at large $N$ in the IR and exhibits maximal chaos, and dual to the Jackiw-Teitelboim 2D gravity) \cite{WittenSYK, TASIKlebanov}. This has driven the development of tensor models in the last few years. New tensor models have been introduced and their large $N$ limits explored \cite{NewModels}. Some models could also be explored beyond leading order using combinatorial techniques \cite{BeyondLO}.

A key feature of tensor models is that the set of observables and interactions is quite larger than in matrix models \cite{Uncoloring}, and grows with $d$. In $U(N)$-invariant matrix models, observables are products of traces $\tr M^n$, for $M$ Hermitian. In $U(N)^2$-invariant models, they are products of traces $\tr(MM^\dagger)^n$, for $M$ a complex matrix. In both cases, there is a single invariant at fixed degree in $M$ (in addition to products of invariants of smaller degree). More generally, there is a set of generators of the ring of $U(N)^d$-invariant polynomials, called the set of bubbles. They are characterized by a $d$-tuple of permutations, up to a left and a right action on the tuple. There is a graphical representation as $d$-regular bipartite graphs with edges labeled by a color from $\{1, \dotsc, d\}$ such that all colors are incident on every vertex. They have been studied in \cite{BenGelounRamgoolam}, where a relation to Kronecker coefficients was found. Enforcing other sets of symmetries leads to other sets of observables, like using $O(N)^d$ instead of $U(N)^d$ relaxes the bipartiteness of the bubbles \cite{O(N)Model,ONinvariants}.

This enlarged set of observables in tensor models compared to matrix models is the source of various universality classes found in the large $N$ and continuum limit. Indeed, it is well-known in 2D that models built with interaction $\tr M^k$, generating $k$-angulations, all have the same universality class (that of pure 2D quantum gravity). However, for $d>2$, there are more possible bubbles, i.e. more interactions at fixed order $k$ in $T, \oT$, which correspond to different $d$-dimensional building blocks. Choosing different bubbles as interactions can then lead to different universality classes \cite{UniversalityClasses}. This however does not seem to be the case in 3D, where all planar bubbles (dual to building blocks homeomorphic to the ball) used as interactions always lead to the universality class of random trees (i.e. branched polymers) \cite{3D}.

\paragraph{Blobbed topological recursion --} 
A natural question for tensor models is to go beyond the large $N$ limit. In particular, it is natural to ask whether the methods used for this purpose for matrix models still work for tensor models and whether it depends on the set of chosen interactions. In this paper we will focus on the topological recursion of Eynard-Orantin \cite{OriginalTR, TR}. Let us nevertheless mention previous works on tensor models beyond the large $N$ limit. A standard combinatorial analysis of maps was applied to the Feynman graphs of the so-called colored tensor models by Gurau and Schaeffer \cite{GurauSchaeffer}, and extended to the set of Feynman graphs of the multiorientable model (which has $U(N)^2\times O(N)$ symmetry, at $d=3$) by Fusy and Tanasa \cite{FusyTanasa}. They classify the Feynman graphs appearing at a given order of the $1/N$ expansion of their respective models. They also identify those which are the most singular in the continuum limit at each order of the $1/N$ expansion, thereby allowing for a double-scaling limit, whose 2-point function was calculated.

In parallel, interest grew around the so-called quartic melonic model. It is a tensor model with up to $d$ quartic interactions having a special structure called melonic. This interest in the quartic model comes from the existence of the Hubbard-Stratonovich technique which transforms this tensor model into a multi-matrix model. It opened up a new way of analyzing tensor models through matrix models. The double-scaling limit for this model was done in \cite{DoubleScalingDartois} (for a result similar to \cite{GurauSchaeffer}). It was also realized in \cite{IntermediateT4} that in the large $N$ limit, the eigenvalues do not spread because of the Coulomb repulsion is subdominant. Instead, they all fall in the potential well, as anticipated in \cite{SDLargeN}. One can then study the fluctuations around the saddle point, an analysis started in \cite{IntermediateT4} where the leading order fluctuations were shown to obey Wigner's semi-circle law.

In \cite{QuarticTR}, the first instance of topological recursion in the context of tensor models was established. Recall that in the ordinary Hermitian 1-matrix model, the topological recursion applies to the calculation of the $n$-point, genus $g$, correlation functions $W_{n,g}(x_1, \dotsc, x_n)$ which appear in the expansion of connected $n$-point functions
\begin{equation}
\langle \tr \frac{1}{x_1-M} \dotsm \tr \frac{1}{x_n-M}\rangle_{\text{conn}} = \sum_{g\geq 0} N^{2-n-2g} W_{n,g}(x_1, \dotsc, x_n).
\end{equation}
In terms of maps, it is a recursion on the generating functions of maps of genus $g$, with $n$ marked faces whose perimeters are tracked by the variables $x_1, \dotsc, x_n$. The topological recursion takes a universal form, and uses a spectral curve as initial data. The spectral curve is determined by the disc and cylinder functions $W_{1,0}(x)$ and $W_{2,0}(x_1, x_2)$.

In \cite{QuarticTR}, the matrix model is the one obtained in \cite{IntermediateT4} for the fluctuations of the eigenvalues around the saddle point. It has $d$ Hermitian matrices $M_1, \dotsc, M_d$ where $M_c$ is said to be of color $c$ and the correlations now need to have the colors of their variables specified,
\begin{equation}
W_{n}(x_1, c_1; \dotsc, x_n, c_n) = \langle \tr \frac{1}{x_1-M_{c_1}} \dotsm \tr \frac{1}{x_n-M_{c_n}}\rangle_{\text{conn}}
\end{equation}
As it turns out, the coupling between the colors is not too strong and a topological recursion be can derived where the spectral curve is a disjoint union of $d$ spectral curves for Gaussian matrix models, with an additional holomorphic term for the cylinder function. This is due to
\begin{description}
\item[condition 1] the $U(N)^d$ symmetry. It implies that the matrices of different colors can only interact through products of traces of different colors. The action is of the form
\begin{equation} \label{Action}
S_N(M_2, \dotsc, M_d) = \sum_{p_1, \dotsc, p_d\geq 0} t_N(p_1, \dotsc, p_d)\ \tr M_{1}^{p_1} \dotsm \tr M_{d}^{p_d}
\end{equation}
\item[condition 2] the $1/N$ expansion. It is such that only the quadratic terms of the action survive the large $N$ limit,
\begin{equation}
S_N(M_1, \dotsc, M_d) \sim_{N\to\infty} N \sum_{c=1}^d a_c \tr M_c^2 + \sum_{c,c'=1}^d b_{cc'} \tr M_c \tr M_{c'}
\end{equation}
(in which sense is explained in the text).
\end{description}
Those two conditions guarantee that an extension of the topological recursion, called the blobbed topological recursion, or rather a multi-colored extension of the latter, holds with the spectral curve being a disjoint union of Gaussian spectral curves, except for $W_{2,0}(x_1, c_1; x_2, c_2)$ which has an additional holomorphic part compared to its usual form.

The blobbed topological recursion was introduced by Borot \cite{BlobbedTR} and further formalized by Borot and Shadrin \cite{BlobbedTR2}. In our context, it applies to matrix models with multi-trace interactions having a topological expansion, i.e. of the form
\begin{equation}
S_N(M) = \sum_{n, h\geq 0} \sum_{p_1, \dotsc, p_n\geq 0} N^{2-n-2h}\ t^{(h)}(p_1, \dotsc, p_n)\ \tr M^{p_1} \dotsm \tr M^{p_n}
\end{equation}
Combinatorially, those types of models generate stuffed maps, defined in \cite{BlobbedTR}. They are maps which are not built by the gluings of disks but as gluings of surfaces of genus $h$ with $n$ boundary components of perimeters $p_1, \dotsc, p_n$. In \cite{QuarticTR} this interpretation survives with an additional coloring of the boundary components.

In the blobbed topological recursion, the recursion for correlation functions still has the same universal term as the ordinary topological recursion, which calculates the singular parts of the correlation functions. In addition, there are now holomorphic contributions \cite{BlobbedTR, BlobbedTR2}. It is also important to keep in mind that the action \eqref{Action} is in fact topological only for $d=4d'+2$, for $d'\in \mathbbm{N}$, \cite{QuarticTR}, meaning that the couplings take the form $t_N(p_1, \dotsc, p_d) = \sum_{h\geq 0} N^{2-d-2h} t^{(h)}(p_1, \dotsc, p_d)$. For other values of $d$, one can do as if the action were topological by absorbing some $N$-dependence into the couplings $t^{(h)}(p_1, \dotsc, p_d)$, and then apply the blobbed topological recursion.

Here we show how to apply this approach to arbitrary $U(N)^d$-invariant models, provided that there are quartic melonic interactions (among others) and some invertibility condition of a quadratic form at large $N$. This revolves around the fact that after some intermediate field techniques and formal integration, such tensor models can always be rewritten as matrix models with a set of $d$ Hermitian matrices satisfying the conditions 1 and 2. Therefore, the correlation functions of these matrices satisfy the blobbed topological recursion, with the same spectral curve as in the quartic melonic model of \cite{QuarticTR}.

Remarkably (and evidently from \cite{BlobbedTR}), the specifics of the model, i.e. the choice of interactions, only contribute to some effective action and not the form of the blobbed topological recursion. In combinatorial terms, the specifics only contribute to the generating functions of the stuffings of the maps. Proving the blobbed topological recursion does not require knowing the explicit effective action, but only that the conditions 1 and 2 are satisfied. In this sense, the blobbed topological recursion is universal in our framework. The only difference between our formulas and those of \cite{QuarticTR} is that the generating functions of the stuffing were explicitly known in \cite{QuarticTR} while their explicit dependence on the couplings constants will be left unknown here (their $N$-dependence is however important).

\paragraph{Method --} 
There are however some technical obstacles to overcome. Arbitrary tensor models can not be directly transformed into matrix models using the Hubbard-Stratonovich transformation which only works for quartic interactions. This first obstacle was overcome in \cite{StuffedWalshMaps} where it was shown that there are still matrix models rewritings. This was proved using a bijection between the Feynman graphs of the tensor model and those of the corresponding matrix model. A second proof was also provided by manipulations of formal integrals (integrals which are only defined as their Feynman series). Here we will repeat this proof, adapting it to go through the second obstacle, which we explain now.

The method of \cite{StuffedWalshMaps} turns a tensor model into a matrix model with complex matrices $M_C$ labeled by subsets of $\{1, \dotsc, d\}$, i.e. $C\subset \{1, \dotsc, d\}$. However, applying the same recipe as in \cite{QuarticTR} requires to have $d$ Hermitian matrices $M_1, \dotsc, M_d$ instead. This is remedied in two times. We first show that it is possible to replace the complex matrices $M_C$ with pairs of Hermitian matrices $(Y_C, \Phi_C)_{C\subset \{1, \dotsc, d\}}$. Then, provided that the quartic melonic interactions are turned on, it is possible to integrate formally over all matrices except $Y_1, \dotsc, Y_d$. In terms of Feynman graphs, this means that one has combinatorial maps with colored edges corresponding to the matrices $Y_1, \dotsc, Y_d$, and everything else is packed in some stuffing of the maps. Keeping in mind the goal of the topological recursion, it is necessary to control the $N$-dependence of the stuffings in terms of their boundary components.

The next step is to observe that all the eigenvalues of $Y_1, \dotsc, Y_d$ fall into some potential wells at large $N$, and move on to the study of their fluctuations. This is where one observes that the conditions 1 and 2 are still satisfied and leads to the blobbed topological recursion. It would be interesting to know whether condition 2 could be removed in general. It is known to be possible in the case of a single matrix model as originally done in \cite{BlobbedTR}. However, in the multi-colored case, it would require a 1-cut (Brown's) lemma for a system of coupled equations with catalytic variables, thus extending the framework of \cite{BM-Jehanne}, which is outside the reach of the present article.

\paragraph{Expectations --}
When discussing the topological recursion in the context of tensor models, there is another natural question to address, which is how to relate the expectations of generic observables on the tensor side to the quantities evaluated via the topological recursion on the matrix side. In \cite{IntermediateT4} for the quartic melonic model, it was shown that the expectations of $\tr M_c^n$ are expectations of Hermite polynomials of some melonic cyclic bubbles on the tensor side. This relation can also be inverted via Hermite polynomials. In \cite{QuarticTR}, the expectations of arbitrary tensor observables (bubbles) were expressed in terms of quantities evaluated by the matrix models, but it involved summing over Wick contractions.

In the present article, we generalize the Hermite polynomial relationship of \cite{IntermediateT4} to arbitrary observables on both the tensor and matrix sides. To express the expectation of a matrix observable in terms of tensorial observables, one has to take derivatives of the potential (which in the case of the quartic melonic model is quadratic, therefore leading to Hermite polynomials). The other way around, i.e. to express the expectation of a tensorial observable in terms of matrix expectations, one has to take derivatives of some effective potential for the matrices $Y_C$s (in the quartic melonic model, it reduces again to a quadratic potential, hence Hermite polynomials), which comes from integrating all the matrices $\Phi_C$s.

\paragraph{Plan --}
In Section \ref{sec:Definitions} we define the tensor models of interest and their multi-matrix equivalent models. Theorems \ref{thm:IntermediateField} and \ref{thm:Expectations} give some of the relationships between the expectations of observables on the tensor and matrix sides. In Section \ref{sec:Effective}, we explain how to formally integrate all matrices except $Y_1, \dotsc, Y_d$, leading to an effective matrix model in Theorem \ref{thm:PartialIntegration}. We use the same technique of formal integration to express the expectations of tensorial observables in terms of matrix expectations in Theorem \ref{thm:BubbleExpectation}. The large $N$ limit of the effective model is discussed and leads to a matrix model for the fluctuations. We study the latter in Section \ref{sec:Blobbed}, by describing the Schwinger-Dyson/loop equations, which can be analyzed along the lines of \cite{BlobbedTR, QuarticTR}. We only present some key aspects which are needed to state the Theorem \ref{thm:BlobbedTR}, about the blobbed topological recursion, since everything works as in \cite{QuarticTR}.


\section{Definition of the tensor and matrix models} \label{sec:Definitions}

\subsection{Bubbles and partition function} \label{sec:PartitionFunction}

Let $d>2$ an integer. For $c=1, \dotsc, d$, we call $E_c \simeq \mathbbm{C}^N$ the space of color $c$. A tensor $T$ of rank $d$ is an object in $\bigotimes_{c=1}^d E_c$ and its elements are denoted $T_{a_1 \dotsb a_d}$, with $a_c = 1, \dotsc, N$ for all $c=1, \dotsc, d$. In tensor models, one is interested in polynomials in the tensor entries which are invariant under the natural action of $U(N)^d$ on $T$ and $\overline{T}$. This group acts as a different copy of $U(N)$ on each color (each index),
\begin{equation}
T \to \bigotimes_{c\in C} U^{(c)}\ T.
\end{equation}
The only way to realize this invariance is to identify the index of a $T$ and a $\overline{T}$ which are in the same position, i.e. have the same color, and sum over the values of that index. This is represented graphically as follows
\begin{equation}
\sum_{a_c=1}^N T_{\dotsb a_c \dotsb} \overline{T}_{\dotsb a_c \dotsb} = \begin{array}{c} \includegraphics[scale=.4]{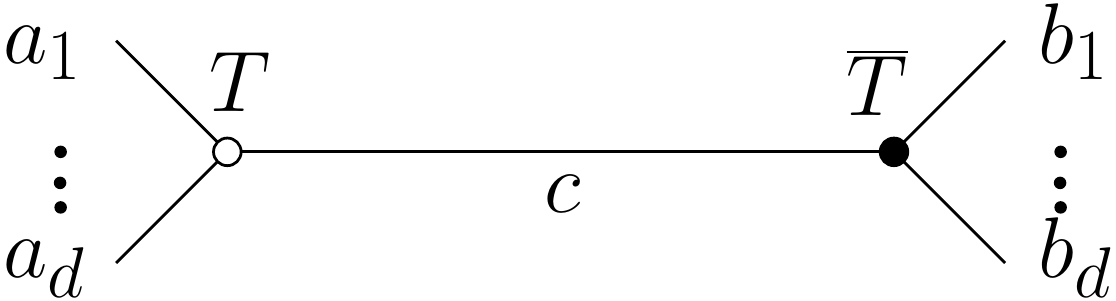}\end{array}.
\end{equation}

A \emph{bubble} is a connected, bipartite graph whose edges are labeled with a color in $\{1, \dotsc, d\}$, and such that each vertex has degree $d$ and all colors are incident to each of them. If $B$ is a bubble, the above rule associates to it a polynomial which is invariant under $U(N)^d$ and denoted $B(T, \overline{T})$. These polynomials generate the ring of $U(N)^d$-invariant polynomials.

If $B$ has $2n$ vertices and one labels the white vertices from $1$ to $n$ and similarly for black vertices, then $B$ can be described as a $d$-tuple $(\tau^{(1)}, \dotsc, \tau^{(d)})$ of permutations on $\{1, \dotsc, n\}$. Set $\tau^{(c)}(v) = v'$ if there is an edge of color $c$ connecting the white vertex $v$ to the black vertex $v'$. The associated polynomial is
\begin{equation} \label{BubblePolynomial}
B(T, \overline{T}) = \sum_{\substack{(i^{(c)}_1, \dotsc, i^{(c)}_n)\\(j^{(c)}_1, \dotsc, j^{(c)}_n)}} \delta^{\tau^{(1)} \dotsb \tau^{(d)}}_{(i^{(c)}_1, \dotsc, i^{(c)}_n), (j^{(c)}_1, \dotsc, j^{(c)}_n)} \prod_{v=1}^n T_{i^{(1)}_v \dotsb i^{(d)}_v} \overline{T}_{j^{(1)}_v \dotsb j^{(d)}_v} \ 
\end{equation}
with by definition
\begin{equation}
\delta^{\tau^{(1)} \dotsb \tau^{(d)}}_{(i^{(c)}_1, \dotsc, i^{(c)}_n), (j^{(c)}_1, \dotsc, j^{(c)}_n)} = \prod_{v=1}^n \prod_{c=1}^d \delta_{i^{(c)}_v, j^{(c)}_{\tau^{(c)}(v)}}
\end{equation}
Invariance under relabeling of the white and black vertices implies invariance of $B(T, \overline{T})$ under left product of $\tau^{(1)}, \dotsc, \tau^{(d)}$ by $\sigma_L$ and right product by $\sigma_R$, two permutations on $\{1, \dotsc, n\}$.

If $C\subset \{1, \dotsc, d\}$ and $\hat{C}$ is its complement, then denote
\begin{equation}
E_C = \bigotimes_{c\in C} E_c \qquad \text{and} \qquad H_{C}(T, \bar{T}) = \begin{array}{c} \includegraphics[scale=.4]{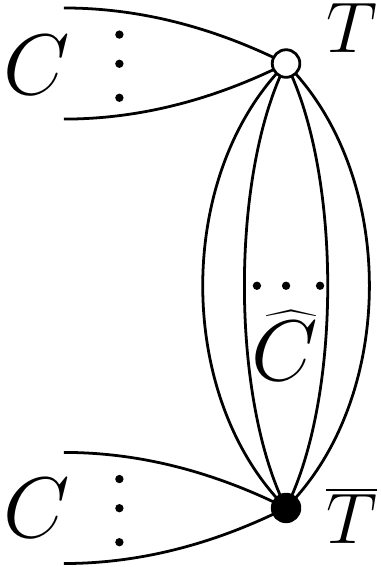} \end{array} \quad \in E_C \otimes E_C^*
\end{equation}
the matrix obtained by contracting all the colors from $\widehat{C}$ between $T$ and $\overline{T}$. We will write $\bigl(H_C(T, \overline{T})\bigr)_{(i^{(c)}), (j^{(c)})}$ the matrix elements.

There is a single quadratic invariant (up to a factor), given by the contraction of $T$ with $\overline{T}$ along all colors,
\begin{equation}
T\cdot\overline{T} = \sum_{a_1,\dotsc, a_d=1}^N T_{a_1\dotsb a_d} \overline{T}_{a_1\dotsb a_d} = H_{\emptyset}(T, \overline{T}).
\end{equation}
For quartic invariants, we choose a color subset $C\subset \{1, \dotsc, d\}$ and connect the indices of $T$ with colors in $C$ with a $\overline{T}$ and those with colors in $\widehat{C}$ with another $\overline{T}$,
\begin{equation}
Q_{C}(T, \overline{T}) = \tr_{E_C} \Bigl(H_C(T, \oT)^2\Bigr) = \begin{array}{c} \includegraphics[scale=.35]{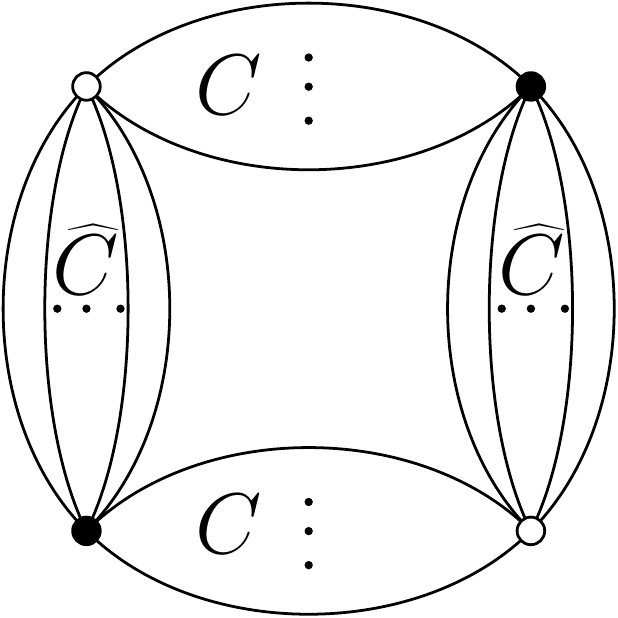}\end{array}
\end{equation}
where the notation $\tr_{E_C}$ indicates that the trace is taken in the spaces with colors in $C$. It is invariant under $C\to \widehat{C}$. In this article, we will furthermore consider cyclic interactions, labeled by a color set $C$ and an integer $n\geq 2$
\begin{equation}
B_{C,n}(T, \overline{T}) = \tr_{E_C} \Bigl(H_{C}(T, \overline{T})^n \Bigr) = \begin{array}{c} \includegraphics[scale=.3]{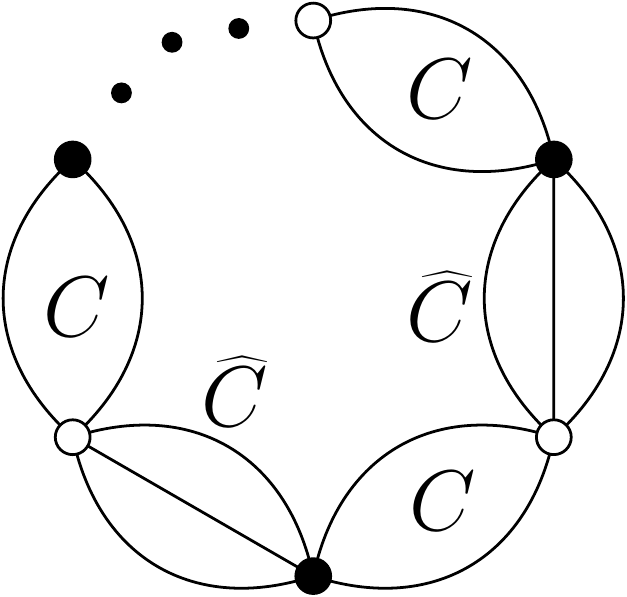} \end{array}
\end{equation}
It is again symmetric under the exchange of $C$ and $\widehat{C}$. We say that the cyclic interaction is \emph{melonic} if $|C|=1$, meaning that in $H_C(T, \overline{T})$, $T$ and $\overline{T}$ are contracted along all colors except one.

Let $I$ be a finite set and $\{B_i\}_{i\in I}$ a finite set of bubbles and $\{t_i\}$ their coupling constants and $\{s_i\}$ some scaling coefficients. Denote $\mathcal{B} = \{(B_i, t_i, s_i)\}_{i\in I}$. Then the partition function is
\begin{equation} \label{TensorModel}
Z_{\Tensor}(N, \mathcal{B}) = \int dT d\oT\ \exp -N^{d-1} T\cdot \oT + V_{N, \mathcal{B}}(T, \oT)
\qquad \text{with} \quad
V_{N, \mathcal{B}}(T, \oT) = \sum_{i\in I} N^{s_i} t_i\, B_i(T, \oT)
\end{equation}
and the free energy is $F(N, \mathcal{B}) = \ln Z_{\Tensor}(N, \mathcal{B})$. Here the measure $dTd\overline{T}$ is proportional to the product of the Lebesgue measures over the tensor entries $\prod_{i_1, \dotsc, i_d} dT_{i_1\dotsb i_d} d\overline{T}_{i_1\dotsb i_d}$, normalized so that
\begin{equation} \label{TensorNormalization}
Z_{\Tensor}(N, \emptyset) = \int dT d\oT\ \exp -N^{d-1} T\cdot \oT = 1.
\end{equation}

Moreover, we only consider \eqref{TensorModel} to make sense by expanding $e^{V_{N, \mathcal{B}}(T, \oT)}$ as a series in $T, \overline{T}$ and integrating each term with the Gaussian weight $e^{-N^{d-1}T\cdot \overline{T}}$. More precisely, write
\begin{equation}
e^{V_{N, \mathcal{B}}(T, \oT)} = \sum_{\{n_i\geq 0\}_{i\in I}} \prod_{i\in I} \frac{1}{n_i!} \bigl(N^{s_i}t_i B_i(T, \overline{T})\bigr)^{n_i}
\end{equation}
and perform the integral at fixed $\{n_i\}$ using the Wick theorem. In more details, we expand $B_i(T, \overline{T})^{n_i}$ as a polynomial in the tensor entries
\begin{equation}
\prod_{i\in I} B_i(T, \overline{T})^{n_i} = \sum_{\substack{a^{(1)}_1, b^{(1)}_1,\dotsc, a^{(1)}_p, b^{(1)}_p\\ \vdots\\ a^{(d)}_1, b^{(d)}_1, \dotsc, a^{(d)}_p, b^{(d)}_p}} \delta^{(\{n_i\})}\Bigl(\{a^{(c)}_q, b^{(c)}_q\}^{c=1, \dotsc, d}_{q=1, \dotsc, p}\Bigr) T_{a^{(1)}_1 \dotsb a^{(d)}_1} \overline{T}_{b^{(1)}_1 \dotsb b^{(d)}_1} \dotsm T_{a^{(1)}_p \dotsb a^{(d)}_p} \overline{T}_{b^{(1)}_p \dotsb b^{(d)}_p}
\end{equation}
by taking a product of \eqref{BubblePolynomial}. The tensor $ \delta^{(\{n_i\})}$ is thus a product of Kroneckers. Here $p$ is the total degree, i.e. if $B_i$ is of degree $p_i$ in $T, \overline{T}$, then $p=\sum_{i\in I} n_i p_i$. Then Wick theorem is applied, and leads to an expansion onto pairings, which are here simply permutations $\sigma\in S_p$ on the set of $p$ elements,
\begin{equation}
\int dT d\oT\ e^{-N^{d-1} T\cdot \oT} T_{a^{(1)}_1 \dotsb a^{(d)}_1} \overline{T}_{b^{(1)}_1 \dotsb b^{(d)}_1} \dotsm T_{a^{(1)}_p \dotsb a^{(d)}_p} \overline{T}_{b^{(1)}_p \dotsb b^{(d)}_p} = N^{-(d-1)p} \sum_{\text{pairings $\sigma\in S_p$}} \delta^{(\sigma)}\Bigl(\{a^{(c)}_q, b^{(c)}_q\}^{c=1, \dotsc, d}_{q=1, \dotsc, p}\Bigr)
\end{equation}
with
\begin{equation}
\delta^{(\sigma)}\Bigl(\{a^{(c)}_q, b^{(c)}_q\}^{c=1, \dotsc, d}_{q=1, \dotsc, p}\Bigr) = \prod_{q=1}^p \prod_{c=1}^d \delta_{a^{(c)}_q, b^{(c)}_{\sigma(q)}}
\end{equation}
A Feynman graph is $G=(\{n_i\}, \sigma)$ and has amplitude
\begin{equation}
A_{N, \mathcal{B}}(G) = N^{-(d-1)p} \prod_{i\in I} \frac{(N^{s_i} t_i)^{n_i}}{n_i!} \sum_{\substack{a^{(1)}_1, b^{(1)}_1,\dotsc, a^{(1)}_p, b^{(1)}_p\\ \vdots\\ a^{(d)}_1, b^{(d)}_1, \dotsc, a^{(d)}_p, b^{(d)}_p}} \delta^{(\{n_i\})}\Bigl(\{a^{(c)}_q, b^{(c)}_q\}^{c=1, \dotsc, d}_{q=1, \dotsc, p}\Bigr) \delta^{(\sigma)}\Bigl(\{a^{(c)}_q, b^{(c)}_q\}^{c=1, \dotsc, d}_{q=1, \dotsc, p}\Bigr)
\end{equation}
Since the tensors $ \delta^{(\{n_i\})}$ and $ \delta^{(\sigma)}$ are products of Kroneckers, and the sums range from 1 to $N$, those sums give $N^{F(G)}$ for some function $F(G)$ which can be given a simple graphical interpretation.

Draw $n_i$ copies of $B_i(T, \overline{T})$ and use $\sigma$ to connect each white vertex (labeled with $q$) to a white vertex (labeled $\sigma(q)$) with an edge. It is customary to give the color 0 to those edges. A bicolored cycle of colors $\{0, c\}$ is a closed path alternating an edge of color 0 and an edge of color $c$. Denote $F_c(G)$ the number of bicolored cycles of colors $\{0,c\}$. Then, by tracking the sequence of index identification along the Kroneckers in the above calculation, it comes that $F(G)$ is the total number of such bicolored cycles, $F(G) = \sum_{c=1}^d F_c(G)$.

Therefore
\begin{equation} \label{FreeEnergyDirectExpansion}
A_{N, \mathcal{B}}(G) = N^{F(G)-(d-1)p + \sum_{i\in I} s_i n_i} \prod_{i\in I} \frac{t_i^{n_i}}{n_i!}
\end{equation}
and the free energy reads
\begin{equation}
F(N, \mathcal{B}) = \sum_{\text{connected $G$}} A_{N, \mathcal{B}}(G)
\end{equation}
The parameters $s_i$ are necessary so that the model has a large $N$ limit which is non-trivial. A non-trivial large $N$ limit is such that $F(N, \mathcal{B})$ appropriately rescaled\footnote{The usual rescaling is $N^{-d}$.} is a non-trivial\footnote{In fact, one usually requires a condition which is a bit stronger: that an infinite number of graphs from the Feynman expansion contributes to $\lim_{N\to\infty} F(N, \mathcal{B})/N^d$. It could be that only a finite number of graphs contribute at large $N$ so that the free energy is a polynomial in the coupling constants. In practice, these two conditions have always been equivalent and we will not discuss these subtleties further.} function of the coupling constants. {\it A priori}, those parameters depend on the whole set $\{B_i\}$. However, all models solved so far are such that $s_i$ is determined by $B_i$ solely. For instance, if $B_i = B_{C,n}$ then $s_i = (|C|-1)n+d-|C|$.

If $P(T, \oT)$ is a polynomial in the tensor entries, its expectation is
\begin{equation}
\langle P(T, \oT)\rangle_{\mathcal{B}} = \frac{1}{Z_{\Tensor}(N, \mathcal{B})} \int dT d\oT\ P(T, \oT)\ \exp -N^{d-1} T\cdot \oT + V_{N, \mathcal{B}}(T, \oT)
\end{equation}

\subsection{Contracted bubbles}

We introduce another representation of $U(N)^d$, this time on matrices. Let $(\Phi_C)_{C\subset \{1,\dotsc,d\}}$ be a set of Hermitian, or complex, matrices labeled by all color subsets (except the empty set), such that $\Phi_C\in E_C\otimes E_C^*$. The action of $U(N)^d$ is
\begin{equation} \label{MatrixU(N)}
\Phi_C \to \bigotimes_{c\in C} U^{(c)}\ \Phi_C \bigotimes_{c\in C} U^{(c)\dagger}
\end{equation}
Bubbles are graphs which can be associated polynomials which generate the ring of $U(N)^d$-invariant polynomials in $T, \oT$. In this representation, the same role is played by \emph{contracted bubbles}.

A \emph{contracted bubble} $P=(B, \pi)$ is obtain from a bubble $B$ by
\begin{itemize}
\item orienting the edges of $B$ from white to black vertices,
\item choosing a pairing $\pi$ of the vertices of $B$ into pairs $\{v, \pi(v)\}$ where $v$ is white and $\pi(v)$ is black,
\item identifying the vertices $v$ and $\pi(v)$ of each pair and removing the loops.
\end{itemize}
Equivalently, i.e. by a trivial bijection, $P$ is a connected graph with oriented edges, each carrying a color $c\in\{1, \dotsc, d\}$, and such that the sub-graph $P_c$, for $c=1, \dotsc, d$, made of the edges of color $c$ only, is a disjoint union of oriented cycles. As a remark, we recall that $P$ can further be transformed into a map with colored edges, as shown in \cite{StuffedWalshMaps}. This will not be necessary here.

There is a bijection between the vertices of $P=(B, \pi)$ and the white vertices of $B$. For this reason we will identify them and use the same notations. From the definition, we see that every vertex $v$ of $P$ carries a color set $C_v \subset \{1, \dotsc, d\}$, with exactly one incoming and one outgoing edges of every color $c\in C_v$.

To obtain an invariant polynomial $P(\{\Phi_C\})$ from the contracted bubble $P$, one associates to every vertex $v_C$ of color set $C$ a matrix $\Phi_C$ and to every incoming incident edge of color $c$ a right index $j^{(c)}_v$ and to every outgoing incident edge a left index $i^{(c)}_v$, e.g.
\begin{equation}
\Bigl(\Phi_{\{c_1, c_2, c_3, c_4\}}\Bigr)_{(i^{(c_1)} i^{(c_2)} i^{(c_3)} i^{(c_4)}), (j^{(c_1)} j^{(c_2)} j^{(c_3)} j^{(c_4)})} = \begin{array}{c} \includegraphics[scale=.5]{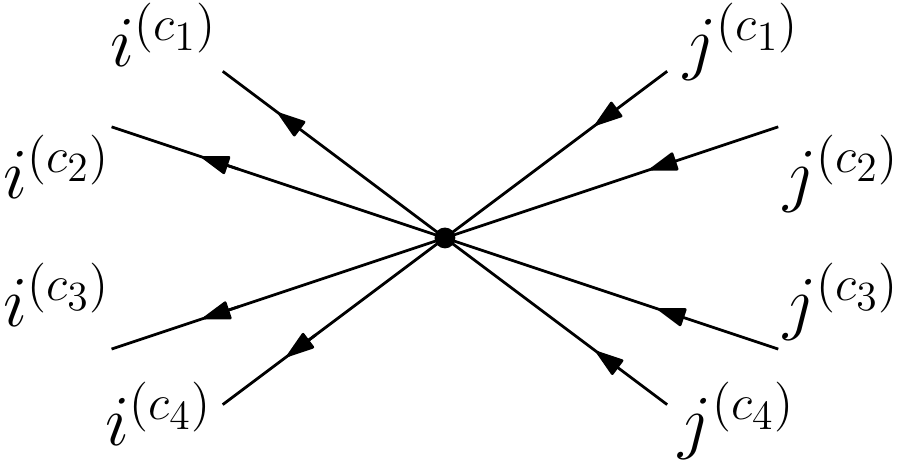} \end{array}
\end{equation}
A left index of color $c$ of a $\Phi_C$ at vertex $v_C$ is identified with the right index of the same color of a $\Phi_{C'}$ at vertex $v_{C'}$ if there is an edge of color $c$ from $v_C$ to $v_{C'}$.

\begin{proposition} \label{thm:BubblePairing}
The polynomial associated to the contracted bubble $P=(B, \pi)$ is related to the bubble polynomial of $B$ as follows
\begin{equation}
P(\{H_C(T, \oT)\}) = B_P(T, \oT).
\end{equation}
\end{proposition}

\begin{proof}
We simply rewrite $B(T, \oT)$ in terms of the matrices $H_C(T, \oT)$ given the choice of pairing $\pi$. Indeed, each white vertex carries a $T$ and each black vertex carries a $\oT$. If $v$ and $\pi(v)$ are connected by edges with colors in $\widehat{C}$, then the sums over the indices of this $T$ and $\oT$ with colors in $\widehat{C}$ form the matrix $H_C(T, \oT)$. Therefore
\begin{equation} \label{BubblePolynomialPairing}
\begin{aligned}
B(T, \overline{T}) &= \sum_{\substack{(i^{(c)}_v)_{c\in C_v}\\(j^{(c)}_v)_{c\in C_v}}}
 \prod_{v=1}^n \biggl[H_{C_v}(T, \overline{T})_{(i^{(c)}_v), (j^{(c)}_{\pi(v)})} \prod_{c\in C_v} \delta_{i^{(c)}_{v}, j^{(c)}_{\tau^{(c)}(v)}} \biggr]\\
&= \sum_{\substack{(i^{(c)}_v)_{c\in C_v}\\(j^{(c)}_v)_{c\in C_v}}} \delta^{\tau^{(1)}\dotsb \tau^{(d)}, \pi}_{(i^{(c)}_v)_{c\in C_v},(j^{(c)}_v)_{c\in C_v}}\ \prod_{v=1}^n H_{C_v}(T, \overline{T})_{(i^{(c)}_v), (j^{(c)}_{v})}
\end{aligned}
\end{equation}
with
\begin{equation} \label{BubbleTensorPairing}
\delta^{\tau^{(1)}\dotsb \tau^{(d)}, \pi}_{(i^{(c)}_v)_{c\in C_v},(j^{(c)}_v)_{c\in C_v}} = \prod_{v=1}^n \prod_{c\in C_v} \delta_{i^{(c)}_{v}, j^{(c)}_{\pi^{-1}\circ\tau^{(c)}(v)}}
\end{equation}
Using the bijection between the white vertices of $B$ and the vertices of $P=(B, \pi)$, we recognize the dependence on $H_C$ as the function $P(\{H_C\})$,
\begin{equation} \label{ContracteBubblePolynomial}
P(\{\Phi_C\}) = \sum_{\substack{(i^{(c)}_v)_{c\in C_v}\\(j^{(c)}_v)_{c\in C_v}}} \delta^{\tau^{(1)}\dotsb \tau^{(d)}, \pi}_{(i^{(c)}_v)_{c\in C_v},(j^{(c)}_v)_{c\in C_v}}\ \prod_{v=1}^n \bigl(\Phi_{C_v}\bigr)_{(i^{(c)}_v), (j^{(c)}_{v})}
\end{equation}
\end{proof}

\subsection{Multi-matrix model and expectation values}

Let $\{P_i\}_{i\in I}$ be a finite set of contracted bubbles and denote $\mathcal{P} = \{(P_i, t_i, s_i)\}_{i\in I}$ and
\begin{equation}
V_{N, \mathcal{P}}(\{\Phi_C\}) = \sum_{i\in I} N^{s_i} t_i\,P_i(\{\Phi_C\}).
\end{equation}
We then define the partition function, for pairs of matrices $\{X_C, \Phi_C\}_{C\subset \{1, \dotsc, d\}}$,
\begin{equation} \label{MatrixModel}
Z_{\MM}(N, \mathcal{P}) = \int \prod_{C\in \{1, \dotsc, d\}} dX_C d\Phi_C\ \exp -\sum_{C\subset \{1, \dotsc, d\}}\tr_{E_C} \bigl(X_C \Phi_C\bigr) + V_{N, \mathcal{P}}(\{\Phi_C\}) - \tr_{\otimes_c E_c} \ln \Bigl( \mathbbm{1} - N^{-(d-1)}\sum_C \tilde{X}_C\Bigr)
\end{equation}
where $\tilde{X}_C = \mathbbm{1}_{V_{\widetilde{C}}} \otimes X_C$ is the lift of $X_C$ to $\otimes_{c=1}^d E_c$ by adding the identity to the colors $c\not\in C$. Here $\MM$ stands for ``multi-matrix''. In order to proceed to the Feynman expansion, the above logarithm has to be expanded as
\begin{equation} \label{LogExpansion}
\begin{aligned}
-\tr_{\otimes_c E_c} \ln \Bigl( \mathbbm{1} - N^{-(d-1)}\sum_C \tilde{X}_C\Bigr) &= \sum_{n\geq 1} \frac{N^{-(d-1)n}}{n}\ \tr_{\otimes_c E_c} \Bigl(\sum_C \tilde{X}_C\Bigr)^n\\
&= \sum_{\text{words $w=C_1 \dotsm C_n$}} \frac{N^{-(d-1)n}}{n} \tr_{\otimes_c E_c} \tilde{X}_{C_1} \dotsm \tilde{X}_{C_n}
\end{aligned}
\end{equation}

There are two possibilities for the integral over $X_C, \Phi_C$, for each $C\subset \{1, \dotsc, d\}$, and the precise form of the Feynman expansion depends on those choices.
\begin{itemize}
\item $X_C= \Phi_C^\dagger$ are complex matrices, adjoint to each other. 
\item $\Phi_C$ is Hermitian and $X_C =-i Y_C$ where $Y_C$ is Hermitian. In this case, one needs the coupling constant of the quartic bubble $Q_{C}$ to be negative. We write it $-t_C/2$ with $t_C> 0$.
\end{itemize}

The equivalence between those two choices is not a given {\it a priori} because they require different Feynman expansions. In the first case, one uses the quadratic term $\tr_{E_C} X_C\Phi_C$ to define the propagator. In the second case however, one cannot use this term since it reads $-i\tr_{E_C}Y_C \Phi_C$ in terms of Hermitian matrices, and this is not positive-definite. This is the reason why we need to add the condition $ t_C> 0 $. The following lemma proves the equivalence we need.

\begin{lemma} \label{thm:ComplexVsHermitian}
For positive coupling constants $t, \tau$, and a potential $U$ which is a series in two variables, the following equality holds formally 
\begin{equation}
\int_{\mathbbm{M}_N(\mathbbm{C})} dZ dZ^\dagger\ e^{-\tr ZZ^\dagger - \frac{t}{2} \tr Z^{2} - \frac{\tau}{2} \tr Z^{\dagger 2} + U(Z,Z^\dagger)} = \int_{\mathcal{H}_N^2} dY d\Phi\ e^{i\tr Y\Phi - \frac{t}{2} \tr \Phi^2 - \frac{\tau}{2} \tr Y^2 + U(\Phi, -iY)}.
\end{equation}
Here $\mathbbm{M}_N(\mathbbm{C})$ is the set of complex $N\times N$ matrices and $\mathcal{H}_N$ the set of $N\times N$ Hermitian matrices.
\end{lemma}

\begin{proof}
\begin{figure}
\includegraphics[scale=.5]{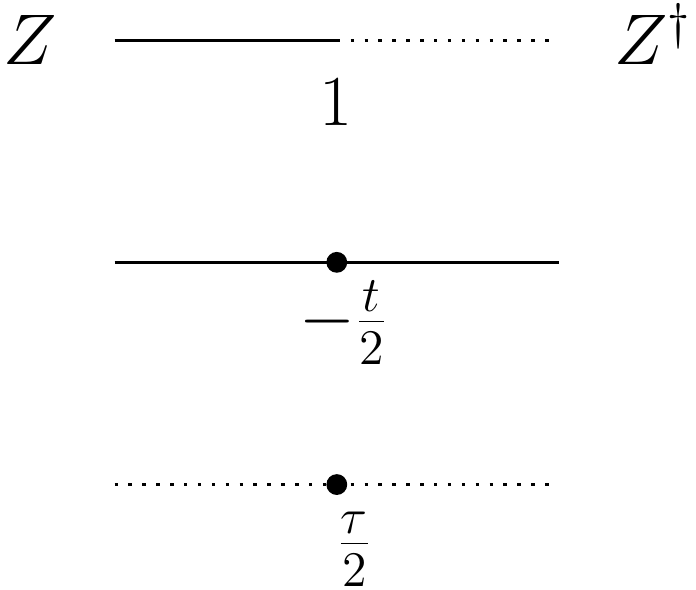} \hspace{2cm} \includegraphics[scale=.5]{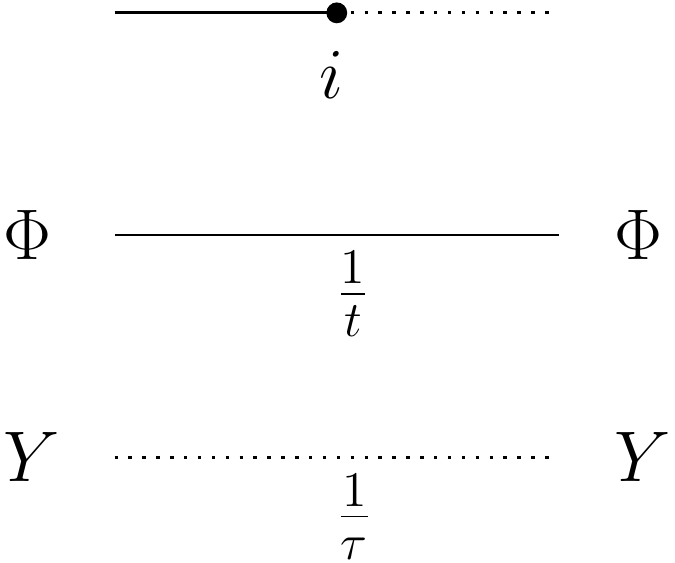}
\caption{\label{fig:FeynmanRules} On the left hand side are the Feynman rules for the complex model, with one propagator and two types of bivalent vertices. On the right hand side are the Feynman rules for the Hermitian model, with two types of propagators and one type of bivalent vertices.}
\end{figure}

The Feynman rules of the left hand side and right hand side, for propagators and bivalent vertices, are in Figure \ref{fig:FeynmanRules}. The Feynman expansion of the left hand side is obtained from the expansion
\begin{equation}
\sum_{l, n, m} \frac{(-t/2)^n (-\tau/2)^m}{l!n!m!} \int_{\mathbbm{M}_N(\mathbbm{C})} dZ dZ^\dagger\ e^{-\tr ZZ^\dagger} \Bigl(\tr Z^{2}\Bigr)^n \Bigl(\tr Z^{\dagger 2}\Bigr)^m U(Z,Z^\dagger)^l
\end{equation}
and performing Wick contractions between $Z$s and $Z^\dagger$s. The Feynman rules are thus
\begin{itemize}
\item Solid half-edges corresponding to the matrix $Z$ and dotted half-edges corresponding to $Z^\dagger$.
\item The propagator, coming from the quadratic term $-\tr ZZ^\dagger$, gives rise to edges which have a solid half and a dotted half, with weight 1.
\item Special vertices of degree 2 with weight $-\frac{t}{2}$ with two incident solid half-edges.
\item Special vertices of degree 2 with weight $-\frac{\tau}{2}$ with two incident dotted half-edges.
\item Other vertices coming from the series expansion of $U(Z, Z^\dagger)$.
\end{itemize}
We call the set of graphs from this expansion $\mathcal{G}_{\text{complex}}$.

The Feynman expansion of the right hand side is obtained from the expansion
\begin{equation}
\sum_{l, p} \frac{i^p}{l!p!} \int_{\mathcal{H}_N^2} dY d\Phi\ e^{- \frac{t}{2} \tr \Phi^2 - \frac{\tau}{2} \tr Y^2} \Bigl(\tr \Phi Y\Bigr)^p U(\Phi, -iY)^l
\end{equation}
and performing independent Wick contractions between pairs of $\Phi$s and between pairs of $Y$s.
\begin{itemize}
\item Solid half-edges corresponding to the matrix $\Phi$ and dotted half-edges corresponding to $X$.
\item Propagators, coming from the quadratic terms $ - \frac{t}{2} \tr \Phi^2 - \frac{\tau}{2} \tr Y^2$, give rise to two types of edges: either two solid half-edges, with weight $1/t$, or two dotted half-edges, with weight $1/\tau$.
\item Special vertices of degree 2 with weight $i$ with an incident solid half-edge and an incident dotted half-edge.
\item Other vertices coming from the series expansion of $U(\Phi, -iY)$.
\end{itemize}
We call the set of graphs from this expansion $\mathcal{G}_{\text{Hermitian}}$.

We show that summing the chains of bivalent vertices in both $\mathcal{G}_{\text{complex}}$ and $\mathcal{G}_{\text{Hermitian}}$ leads to the same new set of rules, for a set of graphs we denote $\mathcal{G}_{\text{summed}}$. These graphs are defined as follows.
\begin{itemize}
\item The have solid and dotted half-edges, and three types of edges: fully solid edges with weight $\tau(t\tau+1)$, fully dotted edges with weight $-t/(t\tau+1)$ and edges made of a solid and a dotted half-edge with weight $1/(t\tau+1)$.
\item Other vertices coming from the series expansion of $U$.
\end{itemize}
In the expansion of $U$, the solid half-edges are associated to the first variable ($Z$ or $\Phi$) and the dotted half-edges to the second variable ($Z^\dagger$ or $-iY$).

The sum of bivalent chains starting and ending on solid half-edges in $\mathcal{G}_{\text{complex}}$ gives
\begin{equation}
{}^{a}_b \begin{array}{c} \includegraphics[scale=.4]{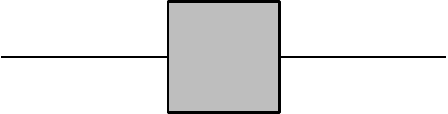} \end{array} {}^{a'}_{b'} = \sum_{n\geq 0} \begin{array}{c} \includegraphics[scale=.6]{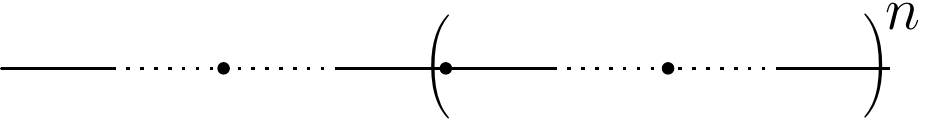} \end{array} = \frac{\tau}{t\tau+1} \delta_{aa'} \delta_{bb'}
\end{equation}
Here the indices $a, b, a', b'$ are the matrix indices which are identified along Wick contractions. Notice that each vertex contributes to either $2\times (-t/2)$ or $2\times \tau/2$ where the extra factors of 2 comes from the two possibilities to add the bivalent vertices, since they are symmetric under the exchange of their incident half-edges. The sum of bivalent chains starting and ending on dotted half-edges is obtained by exchanging $\tau$ with $-t$,
\begin{equation}
{}^{a}_b\begin{array}{c} \includegraphics[scale=.4]{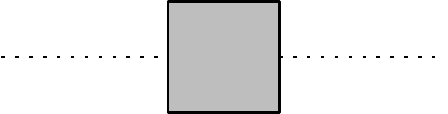} \end{array}{}^{a'}_{b'} = \sum_{n\geq 0} \begin{array}{c} \includegraphics[scale=.6]{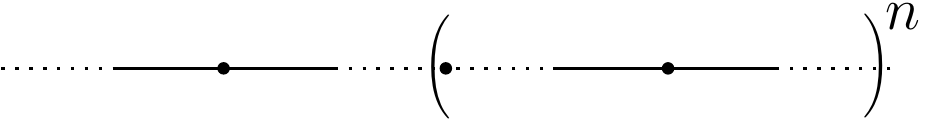} \end{array} = \frac{-t}{t\tau+1} \delta_{aa'} \delta_{bb'}
\end{equation}
The last sum of bivalent chains is between a solid half-edge and a dotted half-edge
\begin{equation}
{}^{a}_b\begin{array}{c} \includegraphics[scale=.4]{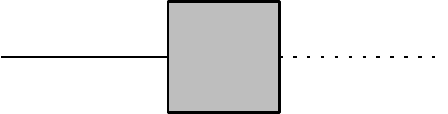} \end{array}{}^{a'}_{b'} = \sum_{n\geq 0} \begin{array}{c} \includegraphics[scale=.6]{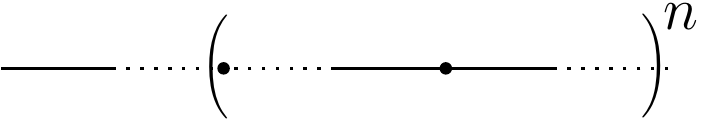} \end{array} = \frac{1}{t\tau+1} \delta_{aa'} \delta_{bb'}
\end{equation}
These are indeed the rules for $\mathcal{G}_{\text{summed}}$.

Performing the same operation in $\mathcal{G}_{\text{Hermitian}}$, one gets
\begin{gather} \label{PhiPhiChain}
{}^{a}_b\begin{array}{c} \includegraphics[scale=.4]{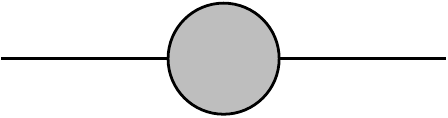} \end{array}{}^{a'}_{b'} = \sum_{n\geq 0} \begin{array}{c} \includegraphics[scale=.6]{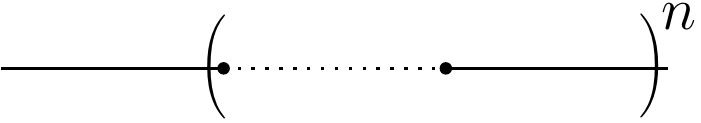} \end{array} = \frac{1}{t}\sum_{n\geq 0} \biggl(\frac{i^2}{t\tau}\biggr)^n = \frac{\tau}{t\tau+1} \delta_{aa'} \delta_{bb'}\\
\label{XXChain}
{}^{a}_b\begin{array}{c} \includegraphics[scale=.4]{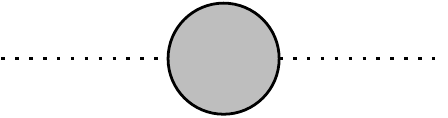} \end{array}{}^{a'}_{b'} = \sum_{n\geq 0} \begin{array}{c} \includegraphics[scale=.6]{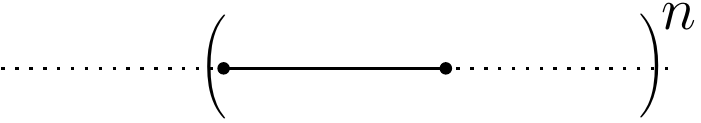} \end{array} = \frac{1}{\tau}\sum_{n\geq 0} \biggl(\frac{i^2}{t\tau}\biggr)^n = \frac{t}{t\tau+1}\delta_{aa'} \delta_{bb'} \\
\label{XPhiChain}
{}^{a}_b\begin{array}{c} \includegraphics[scale=.4]{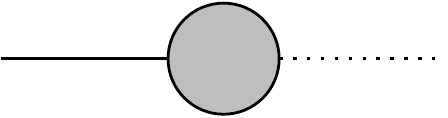} \end{array}{}^{a'}_{b'} = \sum_{n\geq 0} \begin{array}{c} \includegraphics[scale=.6]{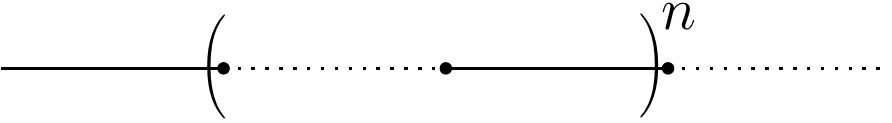} \end{array} = \frac{i}{t\tau}\sum_{n\geq 0} \biggl(\frac{i^2}{t\tau}\biggr)^n = \frac{i}{t\tau+1} \delta_{aa'} \delta_{bb'}
\end{gather}
These are not exactly the expected rules, but the difference is compensated by the other vertices. Indeed, in $\mathcal{G}_{\text{Hermitian}}$, the other vertices come from $U(\Phi, -iY)$, i.e. there is a factor $-i$ on each dotted half-edge incident to such a vertex. Those factors can be re-absorbed so that the weight of the vertices really come from $U(\Phi, Y)$, by multiplying each dotted half-edge in \eqref{PhiPhiChain}, \eqref{XXChain} and \eqref{XPhiChain} by $-i$. This turns those rules into those of $\mathcal{G}_{\text{summed}}$.
\end{proof}

Let $O(\{X_C, \Phi_C\})$ a $U(N)^d$-invariant function (for the simultaneous action \eqref{MatrixU(N)} on $X_C$s and $\Phi_C$s). Its expectation is
\begin{multline} \label{MatrixExpectation}
\langle O(\{X_C, \Phi_C\}) \rangle_{\mathcal{P}} = \frac{1}{Z_{\MM}(N,\mathcal{P})} \int \prod_{C\in \{1, \dotsc, d\}} dX_C d\Phi_C\ O(\{X_C, \Phi_C\})\\
\times \exp -\sum_{C\subset \{1, \dotsc, d\}}\tr_{E_C} \bigl(X_C \Phi_C\bigr) + V_{N, \mathcal{P}}(\{\Phi_C\}) - \tr_{\otimes_c E_c} \ln \Bigl( \mathbbm{1} - N^{-(d-1)}\sum_C \tilde{X}_C\Bigr)
\end{multline}

\begin{lemma} \label{thm:IntermediateField}
Let $f$ be a series which takes as arguments $N\times N$ matrices labeled by the subsets of $\{1, \dotsc, d\}$. Then
\begin{equation} \label{IntermediateField}
\langle f(\{H_C(T, \overline{T})\})\rangle_{V_{N,\mathcal{B}}=0} = \langle f(\{\Phi_C\})\rangle_{V_{N, \mathcal{P}}=0}.
\end{equation}
\end{lemma}
It was proved in \cite{StuffedWalshMaps}, both using a bijection between their Feynman expansions, and using formal integrals in the case of complex variables. Here we briefly reproduce the calculation using formal integrals, in order to later relate the expectations of observables on the tensor and matrix sides using the same technique.

\begin{proof}
Let us first focus on the case where the expectation on the right hand side of \eqref{IntermediateField} is evaluated using complex variables only, $\Phi_C^\dagger = X_C$. Then \eqref{IntermediateField} comes from 
\begin{equation} \label{FundamentalIntegral}
f(h_1, \dotsc, h_m) = \int \prod_{l=1}^m dx_l d\phi_l \ e^{\sum_{l=1}^m (-x_l\phi_l + x_l h_l)}\ f(\phi_1, \dotsc, \phi_m)
\end{equation}
where $\overline{x}_l = \phi_l$ for each $l$, and the integral is over $\mathbbm{C}^m$. It holds via the Wick theorem, order by order in its series expansion.

Making use of \eqref{FundamentalIntegral} with every matrix elements of $H_C, \Phi_C$ as variables, one gets
\begin{equation}
f(\{H_C(T, \overline{T})\}) = \int \prod_C d\Phi_C dX_C\ e^{\sum_C \tr_{E_C} (-X_C\Phi_C + X_C H_C(T,\overline{T}))} f(\{\Phi_C\}) 
\end{equation}
It is now possible to directly integrate the above equation over $T, \overline{T}$ with a Gaussian distribution, leading to
\begin{equation}
\int dTd\overline{T}\ f(\{H_C(T, \overline{T})\})\ e^{-N^{d-1} T\cdot \overline{T}} = \int \prod_C d\Phi_C dX_C\ f(\{\Phi_C\})\ e^{-\sum_C \tr_{E_C} X_C\Phi_C -\tr_{\bigotimes_{c=1}^d E_c}  \ln \bigl( \mathbbm{1} - N^{-(d-1)}\sum_C \tilde{X}_C\bigr)}
\end{equation}
This equality holds up to irrelevant constants. Moreover, the measure on the tensor side has been normalized. On the matrix side, the normalization is trivial when $V_{N, \mathcal{P}}=0$. This proves \eqref{IntermediateField}.

In the case one wishes to use the Hermitian $\Phi_C, Y_C$ with $X_C=-iY_C$, it is necessary to have, instead of vanishing potentials, $V_{N, \mathcal{P}}(\{\Phi_C\}) = - N^{d-1} t_C \tr_{E_C} \Phi_C^2/2$ (and in turn to have on the tensor side a quartic interaction $V_{N, \mathcal{B}}(T, \overline{T}) = -N^{d-1} t_C Q_C(T, \overline{T})/2$). Then Lemma \ref{thm:ComplexVsHermitian} can be applied to turn the integrals over the complex matrix elements to real matrix elements. The coefficient $\tau$ needed in Lemma \ref{thm:ComplexVsHermitian} comes from the expansion of the logarithm in the definition \eqref{MatrixExpectation} of the expectation.
\end{proof}

\begin{theorem} \label{thm:Expectations}
Let $\mathcal{B} = \{(B_i, t_i, s_i)\}_{i\in I}$ as in Section \ref{sec:PartitionFunction} and $\mathcal{P} = \{(P_i, t_i, s_i)\}_{i\in I}$ such that $P_i = (B_i, \pi_i)$. Then,
\begin{equation} \label{PartitionFunctionEquality}
Z_{\Tensor}(N, \mathcal{B}) = Z_{\MM}(N, \mathcal{P}).
\end{equation}
Let $f$ be a series which takes as arguments $N\times N$ matrices labeled by the subsets of $\{1, \dotsc, d\}$. Then
\begin{equation} \label{Expectations}
\langle f(\{H_C(T, \overline{T})\})\rangle_{\mathcal{B}} = \langle f(\{\Phi_C\})\rangle_{\mathcal{P}}.
\end{equation}
\end{theorem}
Again, \eqref{PartitionFunctionEquality} was proved in \cite{StuffedWalshMaps}, both using a bijection, and using formal integrals. 

\begin{proof}
The equality between the partition functions simply derives from Lemma \ref{thm:IntermediateField} for $f(\{\Phi_C\}) = e^{V_{N, \mathcal{P}}(\{\Phi_C\})}$ and the fact that, from Proposition \ref{thm:BubblePairing}, $V_{N, \mathcal{P}}(\{H_C(T, \oT)\}) = V_{N, \mathcal{B}}(T, \oT)$.

This takes care of the denominators in \eqref{Expectations}, and the latter is then equivalent to
\begin{equation}
\langle f(\{H_C(T, \overline{T})\}) e^{V_{N, \mathcal{P}}(\{H_C(T, \overline{T})\})}\rangle_{V_{N,\mathcal{B}} =0} = \langle f(\{\Phi_C\}) e^{V_{N, \mathcal{P}}(\{\Phi_C\})} \rangle_{V_{N,\mathcal{P}}=0}.
\end{equation}
which follows from Lemma \eqref{thm:IntermediateField}.
\end{proof}

\begin{theorem} \label{thm:IF}
With the same notations as previously,
\begin{equation} \label{MatrixToTensor}
\langle P(\{X_C\}) \rangle_{\mathcal{P}} = \sum_{\substack{(i^{(c)}_v)_{c\in C_v}\\(j^{(c)}_v)_{c\in C_v}}} \delta^{\tau^{(1)}\dotsb \tau^{(d)}, \pi}_{(i^{(c)}_v)_{c\in C_v},(j^{(c)}_v)_{c\in C_v}}\ \Big\langle e^{-V_{N, \mathcal{B}}(T, \oT)} \prod_{v=1}^n \frac{-\partial}{\partial \bigl(\Phi_{C_v}\bigr)_{(j^{(c)}_v), (i^{(c)}_v)}} \biggl(e^{V_{N, \mathcal{P}}}\biggr)_{|\Phi_C = H_C(T, \oT)} \Big\rangle_{\mathcal{B}}.
\end{equation}
\end{theorem}

This theorem generalizes \cite{IntermediateT4} in two ways.
\begin{itemize}
\item First, \cite{IntermediateT4} is focused on the quartic melonic model, $V_{N, \mathcal{B}}(T, \oT) = \sum_{c=1}^d -\frac{t_c}{2} N^{d-1} \tr H_c(T, \oT)^2$. In this case, $V_{N, \mathcal{P}}(\{\Phi_C\}) = \sum_{c=1}^d -\frac{t_c}{2} N^{d-1} \tr \Phi_{\{c\}}^2$ is quadratic. The theorem thus explains the appearance of Hermite polynomials in \cite{QuarticTR}. 
\item Theorem \ref{thm:IF} presents the expectation of an arbitrary contracted bubble, while \cite{QuarticTR} only considered $\langle \tr_{E_c} X_\{c\}^n\rangle$.
\end{itemize}	

The reciprocal theorem is Theorem \ref{thm:BubbleExpectation}. It is presented later because it uses the technique of partial integration of Section \ref{sec:PartialIntegral}.

\begin{proof}
To prove \eqref{MatrixToTensor}, we apply Lemma \ref{thm:IntermediateField} to the expectation on the right hand side,
\begin{equation}
\Big\langle e^{-V_{N, \mathcal{B}}(T, \oT)} \prod_{v=1}^n \frac{-\partial}{\partial \bigl(\Phi_{C_v}\bigr)_{(j^{(c)}_v), (i^{(c)}_v)}} \biggl(e^{V_{N, \mathcal{P}}}\biggr)_{|\Phi_C = H_C(T, \oT)} \Big\rangle_{\mathcal{B}} = \Big\langle e^{-V_{N, \mathcal{P}}(\{\Phi_C\})} \prod_{v=1}^n \frac{-\partial}{\partial \bigl(\Phi_{C_v}\bigr)_{(j^{(c)}_v), (i^{(c)}_v)}} \biggl(e^{V_{N, \mathcal{P}}}\biggr) \Big\rangle_{\mathcal{P}}
\end{equation}
which can then be rewritten as
\begin{equation}
\begin{aligned}
&\Big\langle e^{-V_{N, \mathcal{P}}(\{\Phi_C\})} \prod_{v=1}^n \frac{-\partial}{\partial \bigl(\Phi_{C_v}\bigr)_{(j^{(c)}_v), (i^{(c)}_v)}} \biggl(e^{V_{N, \mathcal{P}}}\biggr) \Big\rangle_{\mathcal{P}}\\
&= \frac{1}{Z_{\MM}(N, \mathcal{P})} \int \prod_{C} dX_C d\Phi_C\ \prod_{v=1}^n \frac{-\partial}{\partial \bigl(\Phi_{C_v}\bigr)_{(j^{(c)}_v), (i^{(c)}_v)}} \biggl(e^{V_{N, \mathcal{P}}(\{\Phi_C\})}\biggr) e^{- \sum_C \tr_{E_C}\bigl(X_C \Phi_C\bigr) - \tr_{\otimes_c E_c}\ln\bigl(\mathbbm{1}-N^{-(d-1)}\sum_C \tilde{X}_C\bigr)}\\
&= \frac{1}{Z_{\MM}(N, \mathcal{P})} \int \prod_{C} dX_C d\Phi_C\ \prod_{v=1}^n \frac{\partial}{\partial \bigl(\Phi_{C_v}\bigr)_{(j^{(c)}_v), (i^{(c)}_v)}} \biggl(e^{- \sum_C \tr_{E_C}\bigl(X_C \Phi_C\bigr)}\biggr) e^{V_{N, \mathcal{P}}(\{\Phi_C\}) - \tr_{\otimes_c E_c}\ln\bigl(\mathbbm{1}-N^{-(d-1)}\sum_C \tilde{X}_C\bigr)}\\
&= \frac{1}{Z_{\MM}(N, \mathcal{P})} \int \prod_{C} dX_C d\Phi_C\ \prod_{v=1}^n \bigl(X_{C_v}\bigr)_{(i^{(c)}_v), (j^{(c)}_v)} e^{- \sum_C \tr_{E_C}\bigl(X_C \Phi_C\bigr) + V_{N, \mathcal{P}}(\{\Phi_C\}) - \tr_{\otimes_c E_c}\ln\bigl(\mathbbm{1}-N^{-(d-1)}\sum_C \tilde{X}_C\bigr)}\\
&= \big\langle \prod_{v=1}^n \bigl(X_{C_v}\bigr)_{(i^{(c)}_v), (j^{(c)}_v)} \big\rangle
\end{aligned}
\end{equation}
In the second equality, we have used integration by parts. Summing over all indices with the tensor $\delta^{\tau^{(1)} \dotsb \tau^{(d)}, \pi}$, one recognizes the expansion \eqref{ContracteBubblePolynomial} of $P(\{X_C\})$.
\end{proof}

\section{Effective matrix model} \label{sec:Effective}

Denote $X_c = X_{\{c\}}$ for $c=1, \dotsc, d$ to make the notation lighter. We will show that correlators of the form
\begin{equation}
\langle \tr_{V_{c_1}} \frac{1}{x_{c_1} - \delta X_{c_1}}\ \dotsm\ \tr_{V_{c_n}} \frac{1}{x_{c_n} - \delta X_{c_n}} \rangle_{\text{conn}}
\end{equation}
where $\delta X_c$ is some fluctuation of $X_c$ around its saddle point value, satisfy the blobbed topological recursion. To this aim, we will integrate over all matrices except for $X_1, \dotsc, X_d$ to get an effective matrix model for them.

This effective matrix model has multi-trace interactions. It is convenient to write them using partitions. Given a set of integers $\lambda_1 \geq \dotsb\geq \lambda_l>0$, we say that $\lambda = (\lambda_1, \dotsc, \lambda_l)$ is a partition of length $\ell(\lambda)=l$ and size $|\lambda| = \sum_{i=1}^{\ell{\lambda)}} \lambda_i$. We also allow the special case $\lambda = (0)$ with $\ell(\lambda)=1$. If $M$ is a matrix on $V$, we denote the multi-trace (also known as power sums)
\begin{equation}
I_\lambda(X) = \prod_{i=1}^{\ell(\lambda)} \tr_V X^{\lambda_i}
\end{equation}
Furthermore, we denote $\vlam = (\lambda^{(1)}, \dotsc, \lambda^{(d)})$ a vector of partitions, one for every color, and
\begin{equation}
I_{\vlam}(\{X_c\}) = \prod_{c=1}^d I_{\lambda^{(c)}}(X_c)
\end{equation}
It contains $\ell(\lambda^{(c)})$ traces of color $c$, for a total of $\ell(\vlam) = \sum_{c=1} \ell(\lambda^{(c)})$ traces. It is of total degree $|\vlam| = \sum_{c=1} |\lambda^{(c)}|$ in the matrices.

\begin{proposition} \label{thm:MultiTrace}
A series $F$ in the matrices $X_1, \dotsc, X_d$ which is invariant under $U(N)^d$ has an expansion
\begin{equation}
F(\{X_c\}) = \sum_{\vlam} F(\vlam) I_{\vlam}(\{X_c\}).
\end{equation}
\end{proposition}

\begin{proof}
$F$ has an expansion onto contracted bubbles. The latter are easily seen to be disjoint union of unicycles for all colors, and their corresponding polynomials are $I_{\vlam}$.
\end{proof}

\subsection{Partial integrals} \label{sec:PartialIntegral}

We will assume 
\begin{itemize}
\item The free energy $F(N, \mathcal{B})$ has a $1/N$ expansion
\begin{equation}
F(N, \mathcal{B}) = N^{d'} \sum_{i\geq 0} N^{-\delta_i} F_i(\mathcal{B})
\end{equation}
where $(\delta_i)_{i\geq 0}$ is an increasing, positive sequence and $d'> 2$. As far as we know, all tensor models for which the scaling of the free energy is known satisfy $d'=d$.
\item $I$ is a finite set and $s_i$ is rational for all $i\in I$. This implies rationality of the $\delta_i$s via Equation \eqref{FreeEnergyDirectExpansion}.
\item Each $F_i(\mathcal{B})$ exists in a neighborhood of the origin in the space of coupling constants.
\item Denote $-t_c/2$ the coupling constant of the quartic bubble $Q_{\{c\}}$ for $c=1, \dotsc, d$. We will need $t_c> 0$ in order to choose $\Phi_{\{c\}}$ Hermitian and $X_c= -iY_{c}$ with $Y_{c}$ Hermitian.
\end{itemize}

The results of this section are a bit simplified by rescaling the matrices $X_c$ in the definition \eqref{MatrixModel} by $N^{d-1}$, so that up to irrelevant constants,
\begin{equation}
Z_{\MM}(N, \mathcal{P}) = \int \prod_{C\in \{1, \dotsc, d\}} dX_C d\Phi_C\ \exp -\sum_{C\subset \{1, \dotsc, d\}} N^{d-1} \tr_{E_C} \bigl(X_C \Phi_C\bigr) + V_{N, \mathcal{P}}(\{\Phi_C\}) - \tr_{\otimes_c E_c} \ln \Bigl( \mathbbm{1} - \sum_C \tilde{X}_C\Bigr)
\end{equation}
Denote
\begin{equation}
L(\{X_C\}) = - \tr_{\otimes E_c}\ln\Bigl(\mathbbm{1} - \sum_C \tilde{X}_C\Bigr) - \frac{1}{2} N^{d-1} \sum_{c=1}^d \tr_{E_c} X_c^2
\end{equation}
which just removes some quadratic terms from the expansion of the logarithm (recall $X_c = X_{\{c\}}$).

\begin{theorem} \label{thm:PartialIntegration}
Define the partial free energy 
\begin{equation} \label{PartialFreeEnergy}
\exp F_{N, \mathcal{P}}(\{X_c\}) = 
\int \prod_{\substack{C\subset \{1, \dotsc, d\}\\ |C|\geq 2}} dX_C d\Phi_C \prod_{c=1}^d d\Phi_{\{c\}} e^{-\sum_C N^{d-1} \tr (X_C \Phi_C) + V_{N, \mathcal{P}}(\{\Phi_C\}) + L(\{X_C\})}
\end{equation}
where $X_c=-iY_c$, $Y_c$ Hermitian for $c=1, \dotsc, d$. In the integral $\Phi_{\{c\}}$ are Hermitian, while the pairs $\{\Phi_C, X_C\}$ can be chosen complex, $\Phi_C^\dagger = X_C$, or with $X_C = -iY_C$ with $\Phi_C, Y_C$ Hermitian (provided $t_C>0$) for all $|C|>1$. Then
\begin{equation} \label{EffectiveAction}
F_{N, \mathcal{P}}(\{X_c\}) = \sum_{\vlam} N^{d'-\ell(\vlam)}\ t_{N, \mathcal{P}}(\vlam)\ I_{\vlam}(\{X_c\})
\end{equation}
where $t_{N, \mathcal{P}}(\vlam)$ has a $1/N$ expansion which starts at order $\mathcal{O}(1)$
\begin{equation} \label{1/NExpansionCouplings}
t_{N, \mathcal{P}}(\vlam) = \sum_{i \geq 0} N^{-\delta_i(\vlam)} t_{\mathcal{P}}^{(i)}(\vlam)
\end{equation}
and $(\delta_i(\vlam))_{i\geq 0}$ are positive, increasing sequences of rationals. Furthermore
\begin{equation} \label{NewExpansion}
Z_{\Tensor}(N, \mathcal{B}) = \int \prod_{c=1}^d dX_c\,\exp \biggl(\frac{1}{2} N^{d-1} \sum_{c=1}^d \tr_{E_c} X_c^2 + F_{N, \mathcal{P}}(\{X_c\})\biggr)
\end{equation}
with the same relation between $\mathcal{B}$ and $\mathcal{P}$ as in Theorem \ref{thm:IF}.
\end{theorem}

\begin{proof}
The proof first establishes \eqref{NewExpansion}, by showing that the Feynman expansion of $Z_{\Tensor}(N, \mathcal{B})$ can be obtained from that of $F_{N, \mathcal{P}}(\{X_c\})$. Then, for each Feynman graph of the expansion of $F_{N, \mathcal{P}}(\{X_c\})$, we build a special Feynman graph for $Z_{\Tensor}(N, \mathcal{B})$ in a way such that the $N$-dependence of this construction is controlled. Thus, the $1/N$ expansion of the $Z_{\Tensor}(N, \mathcal{B})$ implies \eqref{EffectiveAction}, \eqref{1/NExpansionCouplings}.

\paragraph{Feynman rules --} Theorem \ref{thm:IF} allows for studying the matrix model instead of the tensor model. Then proving \eqref{NewExpansion} amounts to show that one can write the Feynman rules so as to integrate first over all matrices $\Phi_C, X_C$, except $X_1, \dotsc, X_d$, and then integrate the latter. This will be clear once we have describe how we use the Feynman expansion on $Z_{\MM}(N, \mathcal{P})$.

We choose for the Feynman expansion of $Z_{\MM}(N, \mathcal{P})$ to use $\Phi_{\{c\}}$ Hermitian and $X_c = -iY_c$ with $Y_c$ Hermitian for $c=1, \dotsc, d$. As for the matrix $\Phi_C, X_C$ for $|C|>1$, we can use the complex matrices or the Hermitian matrices either way. Let us choose the latter, since it will let us have a unified description of the Feynman rules, and it is the setup for Theorem \ref{thm:BubbleExpectation}. This however requires $t_C>0$. We thus write integrals over Hermitian matrices,
\begin{equation} \label{RealIntegral}
Z_{\MM}(N, \mathcal{P}) = \int \prod_C d\Phi_C dY_C\ e^{\sum_C \tr_{E_C} \bigl(i N^{d-1} \Phi_CY_C - \frac{t_C}{2} \Phi_C^2 - \frac{N^{d-|C|}}{2} Y_C^2\bigr) + V_{N, \mathcal{P}}(\{\Phi_C\}) + K(\{Y_C\})}
\end{equation}
with
\begin{equation}
\begin{aligned}
K(\{Y_C\}) &= -\tr_{\bigotimes_{c=1}^d E_c} \ln\Bigl(\mathbbm{1} + i \sum_C \tilde{Y}_C\Bigr) + \frac{1}{2} \sum_{C} N^{d-|C|} \tr_{E_C} Y_C^2\\
&= \sum_{\text{words $w\in W$}} \frac{(-i)^{|w|}}{|w|} \tr_{\bigotimes_{c=1}^d E_c} \tilde{Y}_{C_1} \dotsm \tilde{Y}_{C_{|w|}}
\end{aligned}
\end{equation}
which is the series expansion of the logarithm minus its single-trace, quadratic terms (since we have isolated them to be used for the propagators). The set $W$ is a set of words
\begin{equation}
W= \{w=C_1 \dotsm C_{|w|}/\quad \forall q=1, \dotsc, |w| \quad C_q\subset \{1, \dotsc, d\}, \quad \text{$|w|\neq 2$ or $w=C_1 C_2$ with $C_1\neq C_2$}\}
\end{equation}

We have also separated from $V_{N, \mathcal{P}}$ its quadratic terms but we retain the notation. To perform the Feynman expansion, we first expand
\begin{multline}
e^{\sum_C \tr_{E_C} i N^{d-1} \Phi_CY_C  + V_{N, \mathcal{P}}(\{\Phi_C\}) + K(\{Y_C\})} \\
= \sum_{\{n_C\}, \{n_i\}, \{n_w\}} \frac{i^{\sum_Cn_C} (-i)^{\sum_w n_w |w|}\bigl(N^{s_i} t_i\bigr)^{n_i}}{\prod_C n_C! \prod_i n_i! \prod_w n_w! |w|} \prod_{C} \Bigl(\tr_{E_C} \Phi_CY_C\Bigr)^{n_C} \prod_{i} \Bigl(P_i(\{\Phi_C\})\Bigr)^{n_i} \prod_{w}\Bigl(\tr_{\bigotimes_c E_c} \tilde{Y}_{C_1} \dotsm \tilde{Y}_{C_{|w|}}\Bigr)^{n_w}
\end{multline}
Here the indices $C$ span the subsets of $\{1, \dotsc, d\}$, and $i\in I$, and $w\in W$.

Wick theorem can then be applied and, at fixed $\{n_C, n_i, n_w\}$, expresses the Gaussian moments as sums over pairings $\{\sigma_C, \rho_C\}$, the first identifying pairs of $\Phi_C$s and the second identifying pairs of $Y_C$s,
\begin{equation}
\int e^{-\frac{t_C}{2} \tr_{E_C} \Phi_C^2} \Phi_{C a_1 b_1} \dotsm \Phi_{C a_{2n} b_{2n}} = \frac{1}{t_C^n}\sum_{\sigma_C} \prod_{\{i,j\}\in \sigma_C} \delta_{a_i, a_j} \delta_{b_i, b_j}
\end{equation}
where a pairing $\sigma_C$ is a way to partition $\{1, \dotsc, 2n\}$ into disjoint pairs $\{i, j\}$, and
\begin{equation}
\int e^{-\frac{N^{d-|C|}}{2} \tr_{E_C} Y_C^2} Y_{C a_1 b_1} \dotsm Y_{C a_{2n} b_{2n}} = N^{-(d-|C|)n} \sum_{\rho_C} \prod_{\{i,j\}\in \rho_C} \delta_{a_i, a_j} \delta_{b_i, b_j}.
\end{equation}
We call a Feynman graph $G = (\{n_C, \sigma_C, \rho_C\}_{C\subset \{1, \dotsc, d\}}, \{n_i\}_{i\in I}, \{n_w\}_{w\in W})$. The Feynman rules are as follows. 
\begin{itemize}
\item Solid half-edges of color type $C$ correspond to the matrix $\Phi_C$, and dotted half-edges of color type $C$ correspond to $Y_C$.
\item Propagators come from the quadratic terms $-\frac{1}{2} \sum_C  (t_C \tr_{E_C}\Phi_C^2 + N^{d-|C|}\tr_{E_C} Y_C^2)$ and give rise to two types of edges: fully solid edges of color type $C$ with weight $1/t_C$ and fully dotted edges of color type $C$ with weight $N^{|C|-d}$.
\item A special bivalent vertex of weight $iN^{d-1}$ with an incident solid half-edge and an incident dotted half-edge.
\item Vertices coming from the expansion of $K$ have dotted incident half-edges, and vertices coming from the expansion of $V_{N, \mathcal{P}}$ have solid incident half-edges.
\end{itemize}
The sums over all the indices, identified along interactions and Wick contractions will be described below. Denote $\mathcal{G}$ the set of connected Feynman graphs for this Feynman expansion.

It is clear with those rules that one can first perform the integrals over the matrices $\Phi_C, Y_C$ for $|C|>1$ and $\Phi_{\{1\}}, \dotsc, \Phi_{\{d\}}$ by summing over the corresponding pairings, and then integrate over $Y_1, \dotsc, Y_d$ by summing over $\rho_1, \dotsc, \rho_d$. This proves \eqref{NewExpansion}. In fact, one can integrate the matrices $\Phi_C, Y_C$ in any particular order. This means that partial integrals can be performed as one wishes. However, only in our case will we be able to describe the effective action $F_{N, \mathcal{P}}(\{Y_c\})$.

We denote $\mathcal{G}(\{Y_c\})$ the set of connected Feynman graphs for $F_{N, \mathcal{P}}(\{Y_c\})$. They are constructed using the same set of Feynman rules except that fully dotted edges of color $c=1, \dotsc, d$ are not allowed anymore. This means that all dotted half-edges of color $c=1, \dotsc, d$ are left hanging,
\begin{equation}
\begin{array}{c} \includegraphics[scale=.4]{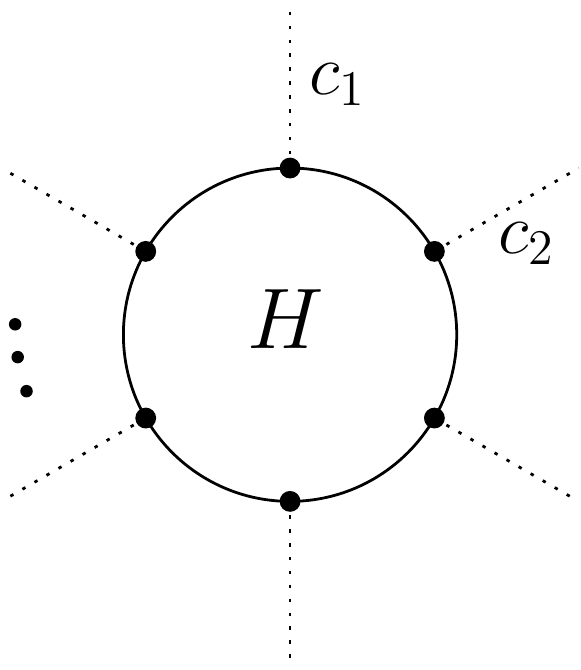} \end{array}
\end{equation}
For instance, in the expansion of $K(\{Y_C\})$, one finds terms with $w=c_1 \dotsb c_n\in\{1, \dotsc, d\}^n$. They contribute to $H_w\in\mathcal{G}(\{Y_c\})$ as graphs with a single vertex and $n$ hanging dotted half-edges of colors $c_1, \dotsc, c_n$ (respecting the cyclic order).

Clearly there is a bijection between $\mathcal{G}$ and collections of graphs from $\mathcal{G}(\{Y_c\})$ connected by Wick pairings $\rho_1, \dotsc, \rho_d$. Graphically, one obtains from $G\in\mathcal{G}$ the collections of graphs from $\mathcal{G}(\{Y_c\})$ by cutting all dotted edges into half-edges. The other way around, performing the Wick pairings $\rho_1, \dotsc, \rho_d$ means connecting the corresponding dotted half-edges.

From the above Feynman rules we have found
\begin{equation}
F_{N, \mathcal{P}}(\{Y_c\}) = \sum_{G\in\mathcal{G}(\{Y_c\})} A_G(N,\mathcal{P}; \{Y_c\})
\end{equation}
where $A_G(N,\mathcal{P}; \{Y_c\})$ is the amplitude of $G$. It is a polynomial in the matrices $Y_c$s (since it is finite object). From the invariance under unitary transformations and Proposition \ref{thm:MultiTrace}, we find that the amplitude of a 
\begin{equation}
F_{N, \mathcal{P}}(\{Y_c\}) = \sum_{\vlam} F_{N, \mathcal{P}}(\vlam) I_{\vlam}(\{Y_c\})
\end{equation}
To establish the $1/N$ expansion of $F_{N, \mathcal{P}}(\vlam)$, we need to look into the structure of the Feynman graphs and their $N$-dependence. The reader already familiar with faces as the origin of the $N$-dependence can skip this discussion.

\paragraph{Faces of Feynman graphs --} As usual in matrix models, a factor $N$ comes from each face, but we need to explain what faces are in this multi-matrix, multi-size context. In ordinary single-trace matrix models, it corresponds to a sequence of identifications of matrix indices via propagators and interactions on a Wick contraction. In terms of Feynman graphs, there is a cyclic order of the half-edges incident to each vertex. This means that Feynman graphs are in fact combinatorial maps. A face is then a closed path which follows an edge to a vertex, then uses the cyclic order around that vertex to move to another half-edge (e.g. counter-clockwise), then follows that edge, etc. 

Consider an index of color $c$ of a $\Phi_C$ in $V_{N, \mathcal{P}}$, or of a $X_C$ in $K(\{X_C\})$, for $c\in C$. In a Feynman graph $G\in\mathcal{G}$, the propagator identifies it with an index of another interaction which has the same color, then because the interaction are unitary invariant, it is further identified with an index of the same color of another matrix, which is then identified to another index of the same color by a propagator, and so on. One ends up with a free sum from 1 to $N$ for each such cycle, hence a power of $N$. We thus see that there could be a notion of faces, but it is color by color, and it requires to track the index identifications in $V_{N, \mathcal{P}}(\{\Phi_C\})$.

One of the key results of \cite{StuffedWalshMaps} is that the interaction $V_{N, \mathcal{P}}(\{\Phi_C\})$ can be given the structure of a map with edges labeled by color type $C\subset \{1, \dotsc, d\}$. Consider $G\in \mathcal{G}$ and denote $G_c\subset G$ for $c=1, \dotsc, d$ the sub-graph whose edges have color set $C\ni c$. It is a disjoint union of ordinary combinatorial maps and the \emph{faces of color $c$} are defined as the faces of those maps. The weight of a graph then goes like $N^{F(G)}$ (and other factors of $N$).


\paragraph{Feynman graphs for $F_{N, \mathcal{P}}(\{Y_c\})$ --} Let $H\in\mathcal{G}(\{Y_c\})$. It has faces of color $c\in\{1, \dotsc, d\}$ as defined above, but some faces go around some hanging dotted half-edges. We call them external faces, while the others are internal faces.

Every internal face contributes with a factor $N$. However, since dotted half-edges of color $c$ correspond to the matrix $Y_c$ which is not integrated over, the Feynman amplitude receives a matrix $Y_c$ every time one goes around a face and meets a hanging dotted half-edge. The matrix indices of $Y_c$ are then identified along the face. An external face thus receives $\tr_{E_c}Y_c^{l_f}$ where $l_f$ is the number of such bivalent interactions around the face.
This shows that for every $H\in\mathcal{G}(\{Y_c\})$, there exists a unique $\vlam$ such that
\begin{equation}
A_H(N,\mathcal{P}; \{Y_c\}) = \tilde{A}_H N^{\eta(G)} I_{\vlam}(\{Y_c\})
\end{equation}
where $\tilde{A}_H$ is independent of $N$ and of the matrices $Y_c$s.

For $H\in\mathcal{G}(\{Y_c\})$, we denote $n_i$ the number of interactions $P_i$, and $b_C$ the number of bivalent interactions $\tr(X_C\Phi_C)$ for $C\subset \{1, \dotsc, d\}$, and $F_{\text{int}}(G)$ the number of external faces. It comes
\begin{equation}
\eta(H) = F_{\text{int}}(H) + \sum_{i\in I} n_i s_i + \sum_C (d-1)b_C.
\end{equation}

We can thus write $\mathcal{G}(\{Y_c\}) = \bigcup_{\vlam} \mathcal{G}_{\vlam}(\{Y_c\})$, where the amplitude of $H\in\mathcal{G}_{\vlam}(\{Y_c\})$ is proportional to $I_{\vlam}(\{Y_c\})$, and
\begin{equation}
F_{N, \mathcal{P}}(\vlam) = \sum_{H\in\mathcal{G}_{\vlam}(\{Y_c\})} \tilde{A}_H N^{\eta(H)}.
\end{equation}
Denote 
\begin{equation}
d_{\vlam} = \sup_{H\in\mathcal{G}_{\vlam}(\{Y_c\})} \eta(H).
\end{equation}
If $d_{\vlam}<\infty$, then $d_{\vlam}$ is actually a maximum. Indeed, we see in Equation \eqref{FreeEnergyDirectExpansion} that the exponent of $N$ is finite sum of integers, except for the $s_i$. But there are a finite number of them and there are rationals and the same for all graphs. Therefore one can write the exponents of $N$ with the same denominator for all graphs, while the numerators consist in a sequence of integers. If the latter has a finite supremum, it is obviously a maximum.

\paragraph{$1/N$ expansion --} A graph $G\in\mathcal{G}$ is made out of sub-graphs $H_{\vlam_1}, \dotsc, H_{\vlam_R} \in\mathcal{G}(\{Y_c\})$ connected by dotted edges of colors in $\{1, \dotsc, d\}$. We denote $E'$ the number of those edges and further denote $F'_c$ the number of faces of color $c$ which go along them, and $F'=\sum_{c=1}^d F'_c$. The sub-graphs $H_{\vlam_1}, \dotsc, H_{\vlam_R} \in\mathcal{G}(\{X_c\})$ come with powers of $N$, $\eta(H_{\vlam_1}), \dotsc, \eta(H_{\vlam_R})$ and amplitudes $\tilde{A}_{H_{\vlam_1}}, \dotsc, \tilde{A}_{H_{\vlam_R}}$. Altogether, the amplitude of $G$ is
\begin{equation} \label{DecomposeDegree}
A_{N, \mathcal{P}}(G) = N^{F' - (d-1)E' + \sum_{r=1}^R \eta(H_{\vlam_r})} \prod_{r=1}^R \tilde{A}_{H_{\vlam_r}}
\end{equation}

Consider $H\in\mathcal{G}_{\vlam}(\{Y_c\})$. We build a graph $G(H)\in\mathcal{G}$ as follows. Let $H_{c}\in \mathcal{G}(\{Y_c\})$ be the graph which consists in a single vertex and a single dotted half-edge of color $c$. Then we connect each hanging dotted half-edge of color $c$ of $H$ to $H_c$,
\begin{equation}
G(H) = \begin{array}{c}\includegraphics[scale=.4]{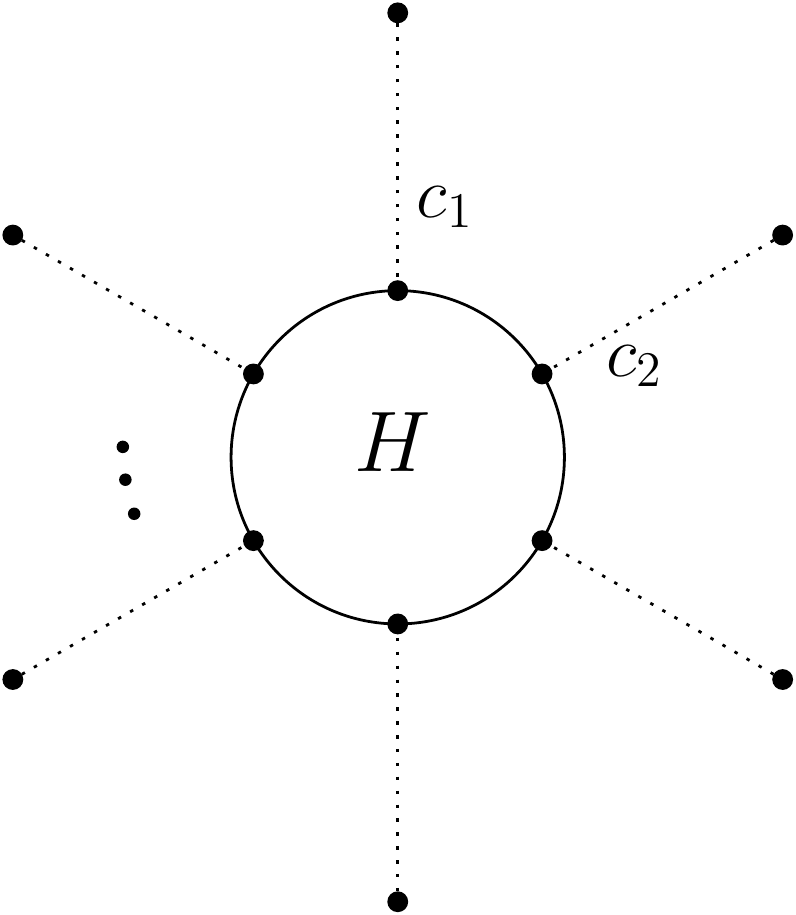} \end{array}
\end{equation}
$H$ corresponds to the interaction $I_{\vlam}$ and therefore has $\ell(\vlam)$ external faces, and $|\vlam|$ hanging dotted half-edges. In $G$, each external face of $H$ becomes an internal face, adding a factor of $N$ to the amplitude. Each univalent vertex $H_c$ is also a connected component for the $(d-1)$ sub-graphs $\mathcal{G}_{c'}$, $c'\neq c$. The total number $F'$ of faces which go along the newly added, fully dotted edges of color $c\in\{1, \dotsc, d\}$ is thus $\ell(\vlam) + (d-1)|\vlam|$. The number $E'$ of those new edges is equal to $I_{\vlam}$. Formula \eqref{DecomposeDegree} thus gives
\begin{equation}
A_{N, \mathcal{P}}(G(H)) = N^{\ell(\vlam) + \eta(H)}\ A(G(H)).
\end{equation}
By assumption, this is bounded by $N^{d'}$ for all $H$. Taking the supremum over all $H\in\mathcal{G}_{\vlam}(\{Y_c\})$, one finds
\begin{equation}
d_{\vlam} \leq d' - \ell(\vlam),
\end{equation}
further implying that $d_{\vlam}$ is a maximum. This shows that $F_N(\vlam) \leq N^{d' - \ell(\vlam)}$ and it is obvious that the sub-leading orders are rational powers of $N$.
\end{proof}

While it is not necessary for our purposes, the leading order of $F(N, \mathcal{P})$ can be described as follows.
\begin{proposition}
The graphs contributing to the leading order of $F(N, \mathcal{P})$ are trees over sub-graphs $H_{\vlam_1}, \dotsc, H_{\vlam_R}$, with $H_{\vlam_r} \in \mathcal{G}_{\vlam_r}(\{Y_c\})$, satisfying $\eta(H_{\vlam_r}) = d'-\ell(\vlam)$.
\end{proposition}

\begin{proof}
We start with the following lemma: an edge $e$ in $G\in\mathcal{G}$ is not a bridge, then $G$ is not leading order. Indeed, if $e$ is of color $c$ and not a bridge, it can be cut into edges of color $c$ as follows without disconnecting $G$,
\begin{equation}
G = \begin{array}{c} \includegraphics[scale=.6]{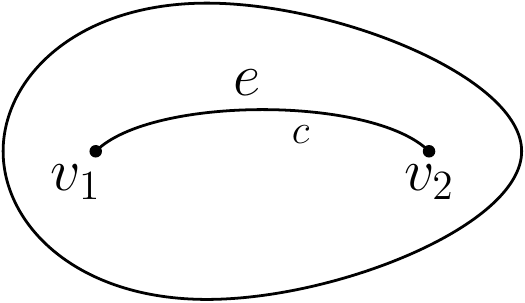} \end{array} \quad \to \quad G' = \begin{array}{c} \includegraphics[scale=.6]{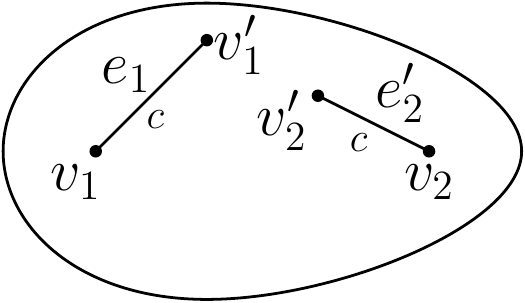} \end{array}
\end{equation}
The exponent of $N$ for $G'$ is calculated from the Feynman rules as follows
\begin{equation} 
A_{N, \mathcal{P}}(H)\ N^{F + \sum_{i\in I} n_i s_i + \sum_C (d-1)b_C - (d-|C|) E_C} A(H)
\end{equation}
for any $H\in\mathcal{G}$, where $A(G)$ is independent of $N$. We now look at the variations of the quantities appearing in this formula. From $G$ to $G'$, the number of edges changes by 1, the number of contracted bubbles $b_C$ of each type is unchanged. The number $F_c$ of faces of color $c$ can increase or decrease by one. Moreover, $\{v'_1\}$ and $\{v_2'\}$ are new connected components of the sub-graphs $G_{c'}$ for $c'\neq c$, implying that the number of faces $F_{c'}$ varies by exactly two. Therefore the total variation of the exponent of $N$ is
\begin{equation}
\Delta\bigl(F + \sum_{i\in I} n_i s_i - (d-1)(E-b)\bigr) = \pm 1 + 2(d-1) - (d-1) \geq  d-2 >0
\end{equation}
which proves the lemma.

As a consequence, a leading order graph is a tree over sub-graphs $H_{\vlam_1}, \dotsc, H_{\vlam_R} \in \mathcal{G}(\{X_c\})$. To maximize the exponent of $N$, we need to maximize each $\eta(H_{\vlam_r})$. As seen in the above proof, $H_{\vlam_r}$ therefore contributes to the leading order if and only if $\eta(H_{\vlam_r}) = d'-\ell(\vlam)$. It is then easy to check that all graphs obtained this way behaves as $N^{d'}$ with respect to $N$ and is thus a leading order graph.
\end{proof}

\subsection{Bubble observables}

In Theorem \ref{thm:IF}, we wrote the expectations of the multi-matrix model in terms of expectations of the tensor model. Here we provide the reverse theorem.

\begin{theorem} \label{thm:BubbleExpectation}
Define the effective action
\begin{equation}
e^{W_{N, \mathcal{P}}(\{X_C\})} = \int \prod_{C\in \{1, \dotsc, d\}} d\Phi_C\ e^{-\sum_{C}\tr_{E_C} \bigl(X_C \Phi_C\bigr) + V_{N, \mathcal{P}}(\{\Phi_C\})}
\end{equation}
with the same techniques as in the proof of Theorem \ref{thm:PartialIntegration} to define $F_{N, \mathcal{P}}(\{X_c\})$. Then for any contracted bubble $P=(B, \pi)$,
\begin{equation} \label{BubbleExpectationMatrix}
\langle B(T, \oT)\rangle_{\mathcal{B}} = \sum_{\substack{(i^{(c)}_v)_{c\in C_v}\\(j^{(c)}_v)_{c\in C_v}}} \delta^{\tau^{(1)}\dotsb \tau^{(d)}, \pi}_{(i^{(c)}_v)_{c\in C_v},(j^{(c)}_v)_{c\in C_v}}\ \Big\langle e^{-W_{N, \mathcal{P}}(\{X_C\})} \prod_{v} \frac{-\partial}{\partial \bigl(X_{C_v}\bigr)_{(j^{(c)}_v), (i^{(c)}_v)}} \biggl(e^{W_{N, \mathcal{P}}(\{X_C\})}\biggr)\Big\rangle_{\mathcal{P}}
\end{equation}
\end{theorem}

It is the reciprocal of Theorem \ref{thm:IF}. It also generalizes \cite{IntermediateT4} in the same two ways: to arbitrary bubbles $B(T, \oT)$ instead of melonic cycles $\tr H_{\{c\}}(T, \oT)^n$, and to an arbitrary model instead of the quartic melonic one. There, since $V_{N, \mathcal{P}}(\{\Phi_C\}) = \sum_{c=1}^d -\frac{t_c}{2} \tr_{E_c} \Phi_{\{c\}}^2$, the integral defining $W_{N, \mathcal{P}}$ can be performed to find $W_{N, \mathcal{P}}(\{X_c\}) = \frac{1}{2t_c} \tr_{E_c} X_c^2$. This is how we recover the Hermite polynomials found in \cite{IntermediateT4}.

Notice that the equivalent of our theorems \ref{thm:IF} and \ref{thm:BubbleExpectation} in \cite{IntermediateT4} both feature Hermite polynomials. Here we see that in general it is not the same object in both theorems, since one has derivatives with respect to $V_{N, \mathcal{P}}(\{\Phi_C\})$, while the other has derivatives with respect to $W_{N, \mathcal{P}}(\{X_C\})$.

\begin{proof}
We start with the following equalities
\begin{equation}
\langle B(T, \oT)\rangle_{\mathcal{B}} = \langle P(\{H_C(T, \oT)\})\rangle_{\mathcal{B}} = \langle P(\{\Phi_C\})\rangle_{\mathcal{P}}
\end{equation}
the first one being Proposition \ref{thm:BubblePairing} and the second Theorem \ref{thm:Expectations}. Then we expand $P$ as a polynomial \eqref{ContracteBubblePolynomial} and use the matrices $X_C$s as sources,
\begin{multline}
\langle \prod_{v=1}^n \bigl(\Phi_{C_v}\bigr)_{(i^{(c)}_v), (j^{(c)}_{v})}\rangle_{\mathcal{P}} 
= \frac{1}{Z_{\MM}(N,\mathcal{P})} \\
\int \prod_{C\in \{1, \dotsc, d\}} dX_C d\Phi_C\ e^{V_{N, \mathcal{P}}(\{\Phi_C\}) - \tr_{\otimes_c E_c} \ln \bigl( \mathbbm{1} - N^{-(d-1)}\sum_C \tilde{X}_C\bigr)} \prod_{v=1}^n \frac{-\partial}{\partial \bigl(X_{C_v}\bigr)_{(j^{(c)}_v), (i^{(c)}_{v})}} e^{-\sum_{C}\tr_{E_C} \bigl(X_C \Phi_C\bigr)}
\end{multline}
then integrate by parts
\begin{multline}
\langle \prod_{v=1}^n \bigl(\Phi_{C_v}\bigr)_{(i^{(c)}_v), (j^{(c)}_{v})}\rangle_{\mathcal{P}} 
= \frac{1}{Z_{\MM}(N,\mathcal{P})} \\
\int \prod_{C\in \{1, \dotsc, d\}} dX_C d\Phi_C\ e^{-\sum_{C}\tr_{E_C} \bigl(X_C \Phi_C\bigr) + V_{N, \mathcal{P}}(\{\Phi_C\})} \prod_{v=1}^n \frac{\partial}{\partial \bigl(X_{C_v}\bigr)_{(j^{(c)}_v), (i^{(c)}_{v})}} e^{ - \tr_{\otimes_c E_c} \ln \Bigl( \mathbbm{1} - N^{-(d-1)}\sum_C \tilde{X}_C\Bigr)}
\end{multline}
At this stage, one performs the integrals over all $\Phi_C$s,
\begin{multline}
\langle \prod_{v=1}^n \bigl(\Phi_{C_v}\bigr)_{(i^{(c)}_v), (j^{(c)}_{v})}\rangle_{\mathcal{P}} \\
= \frac{1}{Z_{\MM}(N,\mathcal{P})} 
\int \prod_{C\in \{1, \dotsc, d\}} dX_C\ e^{W_{N, \mathcal{P}}(\{X_C\})} \prod_{v=1}^n \frac{\partial}{\partial \bigl(X_{C_v}\bigr)_{(j^{(c)}_v), (i^{(c)}_{v})}} e^{ - \tr_{\otimes_c E_c} \ln \Bigl( \mathbbm{1} - N^{-(d-1)}\sum_C \tilde{X}_C\Bigr)}
\end{multline}
and integrate by parts again to find the Theorem.
\end{proof}

\subsection{Comparison with ordinary multi-trace matrix models} \label{sec:MultitraceMatrixModels}

The model \eqref{NewExpansion} with action \eqref{EffectiveAction} has a natural interpretation in terms of stuffed maps as introduced in \cite{BlobbedTR}, with additional colors on their boundary components, and a non-topological expansion.

Recall that a map can be seen as a gluing of polygon along their boundaries. Here a polygon is simply a 2-cell homeomorphic to a disc, with $k$ boundary edges. We then call $k$ the perimeter of the boundary. In stuffed maps, polygons are replaced with \emph{elementary 2-cells of topology $(h,\lambda)$}, where $\lambda = (\lambda_1, \dotsc, \lambda_{\ell(\lambda)})$ is a partition. Such a 2-cell is homeomorphic to a surface of genus $h$ with $\ell(\lambda)$ boundary components of perimeters $\lambda_1, \dotsc, \lambda_{\ell(\lambda)}$. A \emph{stuffed map} is a gluing of elementary 2-cells along the edges of their boundary components.

In matrix models, a polygon of perimeter $k$ corresponds to an interaction $\tr X^k$ in Feynman graphs. In matrix models with multi-trace interactions, whose partition functions are of the form
\begin{equation}
\int dM\ \exp \sum_{\lambda, h} N^{2-\ell(\lambda) - 2h} t^{(h)}(\lambda) I_{\lambda}(X)
\end{equation} 
the interaction $N^{2-\ell(\lambda) - 2h} I_\lambda(X)$ is naturally interpreted as an elementary 2-cell of topology $(h, \lambda)$. Each trace in $I_\lambda(X) = \prod_{i=1}^{\ell(\lambda)} \tr X^{\lambda_i}$ gives rise to a boundary component, and the exponent $\lambda_i$ gives its perimeter.

When interpreting the Feynman expansion of a matrix model with multi-trace interactions in terms of stuffed maps, notice that the only way to associate the genus $h$ to an elementary 2-cell is that the exponent of $N$ in front of $I_{\lambda}(M)$ is $N^{2-\ell(\lambda) - 2h}$ and the coupling constant $t^{(h)}(\lambda)$ is independent of $N$.

In fact, a matrix model with an interaction like 
\begin{equation}
\sum_\lambda N^{2-\ell(\lambda)} t_N(\lambda) I_\lambda(X)
\end{equation}
has a well-defined large $N$ limit when $t_{N}(\lambda)$ itself admits a $1/N$ expansion starting at order $\mathcal{O}(1)$, i.e. $t_N(\lambda) = \sum_{i\geq 0} t^{(i)}(\lambda) N^{-\delta_i(\lambda)}$ where $(\delta_i(\lambda))_{i\geq 0}$ is an increasing sequence of non-negative numbers. However, for the expansions of the free energy and of the correlation functions to be called topological, they must be series in $1/N^2$. It is be the case when the coupling constants $t_N(\lambda)$ are series in $1/N^2$, i.e. $\delta_i(\lambda)\in 2\mathbbm{N}$ but not in general.

The model \eqref{ShiftedPartitionFunction} fit into this framework, with the following amendments.
\begin{itemize}
\item It has $d$ matrices $X_1, \dotsc, X_d$ and interaction in the form $I_{\vlam}(\{X_c\})$. In terms of stuffed maps, it simply means that the boundary components of elementary 2-cells are now colored and a partition $\lambda^{(c)}$ is needed to describe the perimeters of the boundary components of each color $c=1, \dotsc, d$. An elementary 2-cell with boundary profile $\vlam$ moreover comes with the weight $N^{d'-\ell(\vlam)} t_{N, \mathcal{P}}(\vlam)$.
\item Obvious from the above discussion, the $1/N$ expansion is not topological, since $d'\neq 2$, and the sequence $(\delta_i(\vlam))$ in \eqref{1/NExpansionCouplings} may not consist in even integers.
\end{itemize}
In the following section we first focus on the consequence of $d'>2$ for the large $N$ limit.

\subsection{Large \texorpdfstring{$N$} limit and fluctuations}

Theorem \ref{thm:PartialIntegration} provides the form we are looking for to be able to apply the blobbed topological recursion, as in \cite{QuarticTR}. In this section, we thus follow the first step of \cite{QuarticTR} which is to subtract the leading contribution at large $N$. 

\subsubsection{Subtracting the leading order}

Changing variables from $X_c = U_c D_c U_c^\dagger$ to unitary matrices $(U_1, \dotsc, U_d)$ and eigenvalues, $D_c = \operatorname{diag}(x^{(c)}_1, \dotsc, x^{(c)}_N)$ for all $c=1, \dotsc, d$, in the matrix formulation of Theorem \ref{thm:PartialIntegration}, the angular parts are trivially integrated out. The change of variables also produces a squared Vandermonde determinant for all $c=1, \dotsc, d$. This gives
\begin{equation}
Z_{\MM}(N, \mathcal{P}) = \int \prod_{c=1}^d \prod_{i_c=1}^N dx_{i_c}^{(c)}\ \exp \Bigl(\frac{1}{2}N^{d-1} \sum_{c=1}^d \tr_{E_c} D_c^2 + F_{N, \mathcal{P}}(\{D_c\}) + 2\sum_{c=1}^d \sum_{1\leq i_c<j_c\leq N} \ln |x_{i_c}^{(c)} - x_{j_c}^{(c)}| \Bigr)
\end{equation}
If one looks for saddle-points such that the eigenvalues do not scale with $N$, then we see that 
\begin{itemize}
\item the quadratic terms scale like $N^{d}$,
\item all terms from $F_{N, \mathcal{P}}(\{X_c\})$ scale like $N^{d'}$
\item all the terms from Vandermonde determinants scale like $N^2$.
\end{itemize}
This means that we can look for a solution without repulsion between eigenvalues, meaning the eigenvalues $x_i^{(c)}$ can simply fall onto their preferred value for all $c$. We moreover set $d'=d$ so the two first types of contributions have the same scale.

Consider
\begin{equation} \label{SaddlePoint}
X_c = \alpha_c \mathbbm{1}_{E_c} + \frac{1}{N^{\frac{d-2}{2}}} M_c
\end{equation}
The $\alpha_c$s are set on a saddle point of the action for $\{X_c\}$. Moreover, the scaling $1/N^{\frac{d-2}{2}}$ of the fluctuations is chosen so that the leading terms of the action in the $X_c$s scale like $N^2$, i.e. the same as the Vandermonde contributions.

\subsubsection{Matrix model for the fluctuations}

Before plugging \eqref{SaddlePoint} into \eqref{NewExpansion}, let us see its effect on a multi-trace interaction for a single color (which we do not write explicitly),
\begin{equation}
I_\lambda\Bigl(\alpha\mathbbm{1}_V + \frac{1}{N^{\frac{d-2}{2}}} M\Bigr) = \prod_{i=1}^{\ell(\lambda)} \tr_V \Bigl(\alpha\mathbbm{1}_V + \frac{1}{N^{\frac{d-2}{2}}} M\Bigr)^{\lambda_i} = \sum_{\mu_1, \dots, \mu_{\ell(\lambda)}=0}^{\lambda_1, \dotsc, \lambda_{\ell(\lambda)}} \prod_{i=1}^{\ell(\lambda)} \binom{\lambda_i}{\mu_i} \alpha^{\lambda_i-\mu_i}\, N^{-\frac{d-2}{2}\mu_i} \tr_V X^{\mu_i}
\end{equation}
It is easily rewritten as a sum over partitions $\mu\subset \lambda$. We denote the skew-partition $\lambda-\mu = (\lambda_1-\mu_1, \dotsc, \lambda_{\ell(\lambda)}-\mu_{\ell(\lambda)})$ (we can complete $\mu$ with zeros if needed for $\mu_{i\geq \ell(\mu)}$) and $|\lambda-\mu| = \sum_{i\geq 1} \lambda_i-\mu_i$. Also use the notation $\binom{\lambda}{\mu} = \prod_{i=1}^{\ell(\mu)} \binom{\lambda_i}{\mu_i}$. Then
\begin{equation}
I_\lambda\Bigl(\alpha\mathbbm{1}_V + \frac{1}{N^{\frac{d-2}{2}}} M\Bigr) = \sum_{\mu\subset\lambda} \binom{\lambda}{\mu}\,\alpha^{|\lambda-\mu|}\, N^{\ell(\lambda)-\ell(\mu) - \frac{d-2}{2}|\mu|}\ I_\mu(x)
\end{equation}
We thus get the same expansion for $I_{\vlam}(\{\alpha_c \mathbbm{1} + M_c/N^{\frac{d-2}{2}}\})$ by taking a product over the colors. Overall,
\begin{equation}
F_{N, \mathcal{P}}\Bigl(\{\alpha_c \mathbbm{1}_{E_c} + \frac{1}{N^{\frac{d-2}{2}}} M_c\}\Bigr) = \sum_{\vmu} N^{2-\ell(\vmu) - \frac{d-2}{2}(|\vmu|-2)}\, s_{N, \mathcal{P}}(\vmu)\, I_{\vmu}(\{M_c\})
\end{equation}
where the sum over all $d$-tuple of partitions $\vmu = (\mu^{(1)}, \dotsc, \mu^{(d)})$ and $\ell(\vmu) = \sum_{c=1}^d \ell(\mu^{(c)})$ and $|\vmu| = \sum_{c=1}^d |\mu^{(c)}|$. The coefficients are
\begin{equation}
s_{N, \mathcal{P}}(\vmu) = \sum_{\vlam \supset \vmu} t_{N, \mathcal{P}}(\vlam) \prod_{c=1}^d \binom{\lambda^{(c)}}{\mu^{(c)}} \alpha_c^{|\lambda^{(c)} - \mu^{(c)}|}
\end{equation}
They all have $1/N$ expansions starting at order $\mathcal{O}(1)$, simply obtained by using the $1/N$ expansion of the coefficients $t_{N, \mathcal{P}}(\vlam) = \sum_{i \geq 0} N^{-\delta_i(\vlam)} t_{\mathcal{P}}^{(i)}(\vlam)$.

As a result, the matrix model from Theorem \ref{thm:PartialIntegration} becomes
\begin{equation} \label{ShiftedPartitionFunction}
\begin{aligned}
&Z_{\MM}(N, \mathcal{P}) = e^{\frac{1}{2} N^2 \sum_c \alpha_c^2 + F_{N, \mathcal{P}}(\{\alpha_c\mathbbm{1}_{E_c}\})}\, Z_{\text{Fluct}}(N, \mathcal{P})\\ \qquad \text{with} \qquad &Z_{\text{Fluct}}(N, \mathcal{P}) = \int \prod_{c=1}^d dM_c\ \exp \Bigl(\frac{1}{2}N\sum_c \tr_{E_c}M_c^2 + S_{N, \mathcal{P}}(\{M_c\})\Bigr)
\end{aligned}
\end{equation}
where $S_{N, \mathcal{P}}(\{M_c\})$ is
\begin{equation} \label{ShiftedAction}
S_{N, \mathcal{P}}(\{M_c\}) = \sum_{\vmu, |\vmu|\geq 2} N^{2-\ell(\vmu) - \frac{d-2}{2}(|\vmu|-2)} s_{N, \mathcal{P}}(\vmu)\ I_{\vmu}(\{M_c\})
\end{equation}
with
\begin{equation}
s_{N, \mathcal{P}}(\vmu) = \sum_{i\geq 0} N^{-\eta_i(\vmu)} s^{(i)}_{\mathcal{P}}(\vmu)
\end{equation}
and $(\eta_i(\vmu))_{i\geq 0}$ is an increasing sequence of non-negative rationals. The reason there is no linear term in \eqref{ShiftedPartitionFunction} (and no $|\vmu|=1$ term above) is that the set $\{\alpha_c\}$ is a saddle point.

\section{Blobbed topological recursion for colored, multi-trace matrix models} \label{sec:Blobbed}

Going back to the discussion of Section \ref{sec:MultitraceMatrixModels}, but replacing the original model $Z_{\MM}(N, \mathcal{P})$ with the one for the fluctuations, $Z_{\text{Fluct}}(N, \mathcal{P})$ and action \eqref{ShiftedAction}, we see that the latter differs from the multi-trace models of \cite{BlobbedTR} by 
\begin{itemize}
\item the fact that it has $d$ matrices,
\item the fact that the $1/N$ expansions of the coupling constants in \eqref{ShiftedAction} are not topological.
\end{itemize}

To remedy the first difference, consider the following model, which is a generalization of \cite{BlobbedTR} to colored matrices. For a vector $\vell = (\ell_1, \dotsc, \ell_d)$ of non-negative integers, denote $|\vell| = \sum_{c=1}^d\ell_c$. Also denote
\begin{equation}
E_{\vell} = \bigotimes_{c=1}^d E_c^{\otimes \ell_c}
\end{equation}
and $M_c^{(i)} = \mathbbm{1}\otimes \dotsb\otimes M_c\otimes \dotsb\otimes \mathbbm{1}$ the matrix acting on $E_c^{\otimes \ell_c}$ by $M_c$ on the $i$-th factor and the identity everywhere else. Consider the model with partition function
\begin{equation} \label{TopologicalModel}
Z_{\text{Top}}(N, S) = \int \prod_{c=1}^d dM_c\ \exp \sum_{\vell} N^{2-|\vell|} \tr_{E_{\vell}} S_{N, \vell}(M_1^{(1)}, \dotsc, M_1^{(\ell_1)}; \dotsb; M_d^{(1)}, \dotsc, M_d^{(\ell_d)})
\end{equation}
When the action has the following expansion
\begin{equation} \label{TopologicalAction}
S_{N, \vell}(M_1^{(1)}, \dotsc, M_1^{(\ell_1)}; \dotsb; M_d^{(1)}, \dotsc, M_d^{(\ell_d)}) = \sum_{h\geq 0} N^{-2h} S^{(h)}_{\vell}(M_1^{(1)}, \dotsc, M_1^{(\ell_1)}; \dotsb; M_d^{(1)}, \dotsc, M_d^{(\ell_d)})
\end{equation}
it is said to be topological (because it leads to an expansion in $1/N^2$ for the free energy, like the genus expansion of single-trace matrix models). When it has an expansion onto products of traces of the matrices $M_c$s, i.e.
\begin{equation} \label{StuffedMapsAction}
\tr_{E_{\vell}} S_{N, \vell}(M_1^{(1)}, \dotsc, M_1^{(\ell_1)}; \dotsb; M_d^{(1)}, \dotsc, M_d^{(\ell_d)}) = \sum_{\substack{\vlam\\ \ell(\lambda^{(c)}) = \ell_c}} S_N(\vlam) I_{\vlam}(\{M_c\})
\end{equation}
and an invertible quadratic form, then the Feynman expansion corresponds to an expansion onto stuffed maps (non-necessarily topological), as described in Section \ref{sec:MultitraceMatrixModels}. Finally, a topological expansion onto stuffed maps is obtained from
\begin{equation} \label{TopologicalStuffedMapsAction}
\tr_{E_{\vell}} S_{N, \vell}(M_1^{(1)}, \dotsc, M_1^{(\ell_1)}; \dotsb; M_d^{(1)}, \dotsc, M_d^{(\ell_d)}) = \sum_{\substack{h\geq 0\\ \ell(\lambda^{(c)}) = \ell_c}} N^{-2h} S^{(h)}(\vlam) I_{\vlam}(\{M_c\})
\end{equation}
We will present the loop equations of this model. However, we will calculate the disc and cylinder function only in the case where a special property of \eqref{ShiftedAction} is satisfied,
\begin{equation} \label{GaussianReduction}
S^{(0)}(\vlam) = \mathcal{O}(N^{\frac{d-2}{2}}) \qquad \text{for $|\vlam| \geq 3$}
\end{equation}
which in fact will imply that the large $N$ limit is that of a Gaussian model.

The second key difference with \cite{BlobbedTR} is that \eqref{ShiftedAction} is in general non-topological. Since the leading order of the coupling constants is nevertheless $N^{2-\ell(\vlam)}$, as in \eqref{TopologicalModel}. Therefore, we can define new coupling constants as follows
\begin{equation} \label{TopologicalCoefficients}
N^{-2h} S^{(h)}_{N, \mathcal{P}}(\vlam) = \sum_{\substack{i\geq 0\\ 2h = \lfloor \frac{d-2}{2}(|\vlam|-2) + \delta_i(\vlam)\rfloor}} N^{-\frac{d-2}{2}(|\vlam|-2) + \delta_i(\vlam)}\, s^{(i)}_{\mathcal{P}}(\vlam)
\end{equation}
i.e. we absorb $s^{(i)}_{\mathcal{P}}(\vlam)$ and its $N$-dependent prefactor into the coefficients $S^{(h)}_{N, \mathcal{P}}(\vlam)$ by rounding down its order to the closest $2h$.

Here we will follow the same route as in \cite{QuarticTR, BlobbedTR} and do everything as if $S^{(h)}_{N, \mathcal{P}}(\vlam)$ were independent of $N$, except for the explicit evaluation of the disc function and the cylinder function using \eqref{GaussianReduction}. In principle, one also has to eventually re-expand the coefficients $S^{(h)}_{N, \mathcal{P}}(\vlam)$ as in \eqref{TopologicalCoefficients}.

\subsection{Correlation functions}

If $P(\{M_c\})$ is an observable, its expectation is
\begin{equation}
\langle P(\{M_c\}) \rangle = \frac{1}{Z_{\text{Top}}(N, S)} \int \prod_{c=1}^d dM_c\ P(\{M_c\})\ \exp \sum_{\vell} N^{2-|\vell|} \tr_{E_{\vell}} S_{N, \vell}(M_1^{(1)}, \dotsc, M_1^{(\ell_1)}; \dotsb; M_d^{(1)}, \dotsc, M_d^{(\ell_d)})
\end{equation}
The natural set of observables are the expectations of products of traces of the matrices $\{M_c\}$, i.e. expectations of $I_{\vlam}(\{M_c\})$. We introduce the following generating functions for products of $n$ traces of colors $c_1, \dotsc, c_n$,
\begin{equation}
\overline{W}_{n}(x_1, c_1; \dotsc; x_n, c_n) = \Big\langle \prod_{i=1}^n \tr_{V_{c_i}} \frac{1}{x_i - M_{c_i}} \Big\rangle = \sum_{k_1, \dotsc, k_n \geq0} \overline{W}^{(k_1, c_1; \dotsc; k_n, c_n)}_n \prod_{i=1}^n x_i^{-k_i-1}
\end{equation}
i.e.
\begin{equation}
\overline{W}^{(k_1, c_1; \dotsc; k_n, c_n)}_n = \Bigl[ \prod_{i=1}^n x^{-k_i-1}\Bigr] \overline{W}_{n}(x_1, c_1; \dotsc; x_n, c_n) = \Big\langle \prod_{i=1}^n \tr_{E_{c_i}} M_{c_i}^{k_i}\Big\rangle
\end{equation}
and their connected counterparts
\begin{equation}
W_{n}(x_1, c_1; \dotsc; x_n, c_n) = \Big\langle \prod_{i=1}^n \tr_{V_{c_i}} \frac{1}{x_i - M_{c_i}} \Big\rangle_{\text{conn}} = \sum_{k_1, \dotsc, k_n \geq0} W^{(k_1, c_1; \dotsc; k_n, c_n)}_n \prod_{i=1}^n x_i^{-k_i-1}
\end{equation}
i.e.
\begin{equation}
W^{(k_1, c_1; \dotsc; k_n, c_n)}_n = \Bigl[ \prod_{i=1}^n x^{-k_i-1}\Bigr] W_{n}(x_1, c_1; \dotsc; x_n, c_n) = \Big\langle \prod_{i=1}^n \tr_{E_{c_i}} M_{c_i}^{k_i}\Big\rangle_{\text{conn}}
\end{equation}
The variable $x_i$ is said to be of color $c_i$ when it is the generating parameter for $\tr_{E_{c_i}} M_{c_i}^{k_i}$ expanded around infinity. We denote $\C_c$ the copy of $\C$ of color $c$, so that $x_i\in U_{c_i}\C_{c_i}$ for some open subset of $\C_{c_i}$. 

We will also need the functions
\begin{equation}
W_{n}^{(k_1, c_1'; \dotsc; k_l, c_l')}(x_1, c_1;\dotsc; x_{n-l}, c_{n-l}) = \Big\langle \prod_{i=1}^l \tr_{E_{c'_i}} M_{c'_i}^{k_i}\ \prod_{j=1}^{n-l} \tr_{E_{c_j}} \frac{1}{x_j-M_{c_j}} \Big\rangle
\end{equation}
which are obtained from $W_n(x_1, c_1; \dotsc; x_{n-l}, c_{n-l}; x'_1, c_1'; \dotsc; x'_l, c'_l)$ by extracting some series coefficients 
\begin{equation}
W_{n}^{(k_1, c_1'; \dotsc; k_l, c_l')}(x_1, c_1;\dotsc; x_{n-l}, c_{n-l}) = \Bigl[ \prod_{i=1}^l {x'}_i^{-k_i-1}\Bigr] W_n(x_1, c_1; \dotsc; x_{n-l}, c_{n-l}; x'_1, c_1'; \dotsc; x'_l, c'_l)
\end{equation}
As a special case of such functions, when $\vlam = (\lambda^{(1)}, \dotsc, \lambda^{(d)})$ is a $d$-tuple of partitions, and $c\in\{1, \dotsc, d\}$ and $j\in \{1, \dotsc, \ell(\lambda^{(c)})\}$, we denote $\vlam_{(c,j)} = ({\lambda'}^{(1)}, \dotsc, {\lambda'}^{(d)})$ the $d$-tuple of partitions with 
\begin{equation}
{\lambda'}^{(c')} = \lambda^{(c')} \quad \text{for $c'\neq c$ and }\quad {\lambda'}^{(c)} = \bigl(\lambda^{(c)}_1, \dotsc, \lambda^{(c)}_{j-1}, \lambda^{(c)}_{j+1}, \dotsc, \lambda^{(c)}_{\ell(\lambda^{(c)})}\bigr)
\end{equation}
i.e. the $j$-th row of $\lambda^{(c)}$ is removed. Then we denote
\begin{equation}
W_n^{(\vlam_{(c,j)})}(x_1, c_1; \dotsc; x_p, c_p) = W_n^{(\lambda^{c'}_i, c')_{c'=1, \dotsc, d\text{ and } i=1, \dotsc, \ell(\lambda^{(c')})}^{(c',i)\neq (c, j)}}(x_1, c_1; \dotsc; x_p, c_p)
\end{equation}
with $n = p+\ell(\vlam)-1$.

It will also appear natural to introduce \emph{global} correlation functions which are defined on (an open subset of)
\begin{equation}
E_n = \Bigl(\bigcup_{c=1}^d \C_c\setminus \Gamma_c\Bigr)^n
\end{equation}
so that each $x_i$ can be evaluated on any color. These correlation functions are
\begin{equation}
W_n (x_1, \dotsc, x_n) = \sum_{c_1, \dotsc, c_n=1}^d W_{n}(x_1, c_1; \dotsc; x_n, c_n)\, \prod_{i=1}^n \mathbbm{1}(x_i, c_i).
\end{equation}
where $\mathbbm{1}(x,c)$ is 1 if $x\in\C_c$ and 0 else. In terms of components
\begin{equation}
W_n (x_1, \dotsc, x_n) = \sum_{\substack{k_1, \dotsc, k_n\geq 0\\ c_1, \dotsc, c_n=1, \dotsc, d}} W_{n}^{(k_1, c_1; \dotsc; k_n, c_n)}\, \prod_{i=1}^n x_i^{-k_i-1} \mathbbm{1}(x_i, c_i).
\end{equation}
The correlation functions $W_{n}(x_1, c_1; \dotsc; x_n, c_n)$ are said to be the \emph{local expressions} of $W_n(x_1, \dotsc, x_n)$, since each variables is assigned a fixed color. 

\subsection{Loop equations}

In this section we use the form \eqref{StuffedMapsAction} of the action.

\subsubsection{1-point equation}

The Schwinger-Dyson/loop equations are obtained from
\begin{equation}
\frac{1}{Z_{\text{Top}}(N, S)}\int \prod_{c=1}^d dM_c \sum_{a,b=1}^N \frac{\partial}{\partial (M_c)_{ab}}\Bigl( \bigl(M_c^n\bigr)_{ab}\ e^{\sum_{\vlam} N^{2-\ell(\vlam)} S_N(\vlam) I_{\vlam}(\{M_c\})}\Bigr) = 0.
\end{equation}
by making the action of the derivative above explicit on each term of integrand, and summing over $n\geq0$ with $x^{-n-1}$. One gets
\begin{equation} \label{1PointEquation}
\overline{W}_2(x,c;x,c) + \sum_{\vlam} N^{2-\ell(\vlam)} S_N(\vlam) \sum_{j=1}^{\ell(\lambda^{(c)})} \lambda^{(c)}_j \Big\langle \tr_{E_c} \frac{M_c^{\lambda^{(c)}_j-1}}{x-M_c} \prod_{i\neq j} \tr_{E_c} M_c^{\lambda^{(c)}_i}\ \prod_{c'\neq c} I_{\lambda^{(c')}}(M_{c'}) \Big\rangle = 0
\end{equation}

We now work towards rewriting \eqref{1PointEquation} in terms of connected correlation functions. Denote $\{R_1, R_2, \dotsc\}$ a partition of $\{1, \dotsc, n\}$, then
\begin{equation}
\overline{W}_n(x_1, c_1; \dotsc; x_n, c_n) = \sum_{\{R_1, R_2, \dotsc\}} \prod_\alpha W_{|R_\alpha|}(\{x_{R_\alpha}, c_{R_\alpha}\})
\end{equation}
with the short hand notation $\{x_{R_\alpha}, c_{R_\alpha}\} = \{x_r, c_r\}_{r\in R_\alpha}$. This first gives
\begin{equation}
\overline{W}_2(x,c;x,c) = W_1(x,c)^2 + W_2(x,c;x,c)
\end{equation}

The contribution of the interaction is split in the usual way using
\begin{equation}
\frac{M_c^{\lambda^{(c)}_j-1}}{x-M_c} = \frac{x^{\lambda^{(c)}_j-1}}{x-M_c} + \sum_{q=0}^{\lambda^{(c)}_j-2} x^{\lambda^{(c)}_j-2-q} M_c^q
\end{equation}
which leads to 
\begin{equation}
\Big\langle \tr_{E_c} \frac{M_c^{\lambda^{(c)}_j-1}}{x-M_c} \prod_{i\neq j} \tr_{E_c} M_c^{\lambda^{(c)}_i}\ \prod_{c'\neq c} I_{\lambda^{(c')}}(M_{c'}) \Big\rangle 
= x^{\lambda^{(c)}_j-1} \overline{W}_{\ell(\vlam)}^{(\vlam_{(c,j)})}(x,c)
- \sum_{q=1}^{\lambda^{(c)}_j-2} x^{\lambda^{(c)}_j-2-q} \overline{W}_{\ell(\vlam)}^{(\vlam_{(c,j)};q,c)}
\end{equation}
Then rewrite each of the two terms using connected correlations. To do so, denote
\begin{equation}
L(\vlam) = \{(c, i); c\in\{1, \dotsc, d\} \text{ and } i\in\{1, \dotsc, \ell(\lambda^{(c)})\}
\end{equation}
and $\mathcal{P}(L(\vlam))$ the set of partitions of $L(\vlam)$, i.e. $R=\{R_1, \dotsc, R_{\ell(R)}\}\in\mathcal{P}(L(\vlam))$ if the $R_\alpha$s are non-empty, disjoint, and $\mathop{\dot\bigcup}_{\alpha} R_{\alpha} = L(\vlam)$. Moreover, for a fixed pair $(c, j)\in L(\vlam)$, we denote $R(c,j)$ the part which contains $(c, j)$, and
\begin{equation}
R(c,j) = \{(c,j)\}\cup R'(c,j)
\end{equation}
where $R'(c,j)$ can be empty. Then, 
\begin{equation}
\overline{W}_{\ell(\vlam)}^{(\vlam_{(c,j)})}(x,c) = \sum_{R\in\mathcal{P}(L(\vlam))} W_{|R(c,j)|}^{(\lambda_{R'(c,j)}, c_{R'(c,j)})}(x,c) \prod_{\substack{\alpha\\ R_{\alpha}\neq R(c,j)}} W_{|R_\alpha|}^{(\lambda_{R_\alpha}, c_{R_\alpha})}
\end{equation}
Here we use the short-hand notation $(\lambda_{R_\alpha}, c_{R_\alpha})=(\lambda^{(c')}_{i}, c')_{(c',i)\in R_\alpha}$. Furthermore
\begin{equation}
\overline{W}_{\ell(\vlam)}^{(\vlam_{(c,j)};q,c)} = \sum_{R\in\mathcal{P}(L(\vlam))} W_{|R(c,j)|}^{(\lambda_{R'(c,j)}, c_{R'(c,j)}; q,c)} \prod_{\substack{\alpha\\ R_{\alpha}\neq R(c,j)}} W_{|R_\alpha|}^{(\lambda_{R_\alpha}, c_{R_\alpha})}
\end{equation}

The loop equation \eqref{1PointEquation} then reads
\begin{multline}
W_1(x,c)^2 + W_2(x,c;x,c) + \sum_{\substack{\vlam\\ R\in\mathcal{P}(L(\vlam))}} \sum_{j=1}^{\ell(\lambda^{(c)})} N^{2-\ell(\vlam)} S_N(\vlam) \prod_{\substack{\alpha\\ R_{\alpha}\neq R(c,j)}} W_{|R_\alpha|}^{(\lambda_{R_\alpha}, c_{R_\alpha})}\\
\times \lambda^{(c)}_j \Bigl(x^{\lambda^{(c)}_j-1} W_{|R(c,j)|}^{(\lambda_{R'(c,j)}, c_{R'(c,j)})}(x,c) - \sum_{q=0}^{\lambda^{(c)}_j-2} x^{\lambda^{(c)}_j-2-q} W_{|R(c,j)|}^{(\lambda_{R'(c,j)}, c_{R'(c,j)}; q,c)}\Bigr)  = 0
\end{multline}
which can also be written as a global equation
\begin{multline} \label{Global1PointEquation}
W_1(x)^2 + W_2(x,x) + \sum_{\substack{\vlam\\ R\in\mathcal{P}(L(\vlam))}} \sum_{c=1}^d \mathbbm{1}(x,c) \sum_{j=1}^{\ell(\lambda^{(c)})} N^{2-\ell(\vlam)} S_N(\vlam) \prod_{\substack{\alpha\\ R_{\alpha}\neq R(c,j)}} W_{|R_\alpha|}^{(\lambda_{R_\alpha}, c_{R_\alpha})}\\
\times \lambda^{(c)}_j \Bigl(x^{\lambda^{(c)}_j-1} W_{|R(c,j)|}^{(\lambda_{R'(c,j)}, c_{R'(c,j)})}(x,c) - \sum_{q=0}^{\lambda^{(c)}_j-2} x^{\lambda^{(c)}_j-2-q} W_{|R(c,j)|}^{(\lambda_{R'(c,j)}, c_{R'(c,j)}; q,c)}\Bigr)  = 0
\end{multline}

\subsubsection{\texorpdfstring{$n$}-point equations}

The single-trace terms of the potential are
\begin{equation}
\sum_{c=1}^d \sum_{\lambda\geq 1} N \bigl(s^{(0)}(\lambda, c) + o(1)\bigr) \tr_{E_c}M_c^\lambda
\end{equation}
Then the loop insertion operator with respect to color $c\in\{1, \dotsc, d\}$ is
\begin{equation}
\delta_{x} = \sum_{c=1}^d \sum_{\lambda\geq 0} \mathbbm{1}(x,c) x^{-\lambda-1} \frac{\partial}{\partial s^{(0)}(\lambda, c)}
\end{equation}
Furthermore, for integers $A, n$, denote $\mathcal{I}_A(n)$ the set of lists of the form $I=(I_1, \dotsc, I_A)$, such that 
\begin{itemize}
\item $I_\alpha\subseteq \{2,\dotsc, n\}$ is possibly empty, 
\item the non-empty $I_\alpha$s are disjoint,
\item $\bigcup_{\alpha=1}^A I_\alpha = \{2, \dotsc, n\}$
\end{itemize}

Repeated actions of the loop insertion operator on \eqref{Global1PointEquation} gives
\begin{multline} \label{nPointEquation}
\sum_{(I_1, I_2)\in \mathcal{I}_2(n)} W_{|I_1|+1}(x_1, x_{I_1}) W_{|I_2|+1}(x_1,x_{I_2}) + W_{n+1}(x_1, x_1, \dotsc, x_n) \\
+ \sum_{j=2}^n \mathbbm{1}(x_1,x_j) \frac{\partial}{\partial x_j} \frac{W_{n-1}(x_2, \dotsc, x_n) - W_{n-1}(x_2, \dotsc, x_{j-1}, x_1, x_j, \dotsc, x_n)}{x_j-x_1}\\
+ \sum_{\substack{\vlam\\ R\in\mathcal{P}(L(\vlam))}} \sum_{(I_1, \dotsc, I_{\ell(R)}) \in \mathcal{I}_{\ell(R)}(n)} \sum_{c=1}^d \mathbbm{1}(x_1,c) \sum_{j=1}^{\ell(\lambda^{(c)})} N^{2-\ell(\vlam)} S_N(\vlam) \prod_{\substack{\alpha\\ R_{\alpha}\neq R(c,j)}} W_{|R_\alpha|+ |I_\alpha|}^{(\lambda_{R_\alpha}, c_{R_\alpha})}(x_{I_\alpha})\\
\times \lambda^{(c)}_j \Bigl(x_1^{\lambda^{(c)}_j-1} W_{|R(c,j)|+|I(c,j)|}^{(\lambda_{R'(c,j)}, c_{R'(c,j)})}(x_1,x_{I(c,j)}) - \sum_{q=0}^{\lambda^{(c)}_j-2} x_1^{\lambda^{(c)}_j-2-q} W_{|R(c,j)|+|I(c,j)|}^{(\lambda_{R'(c,j)}, c_{R'(c,j)}; q,c)}(x_{I(c,j)}\Bigr)  = 0
\end{multline}
$I(c,j)$ is defined as $I_{\alpha^*}$ where $\alpha^*$ is the index such that $R_{\alpha^*} = R(c, j)$. Moreover, $\mathbbm{1}(x,y) = \sum_{c=1}^d \mathbbm{1}(x,c) \mathbbm{1}(y,c)$ is 1 if and only if $x$ and $y$ are variables of the same colors.

\subsubsection{Topological expansion}

All correlation functions admit the usual topological expansion
\begin{equation}
W^{(k_1, c_1; \dotsc; k_p, c_p)}_{n+p}(x_1, \dotsc, x_n) = \sum_{h\geq 0} N^{2-n+p-2h}\, W^{(k_1, c_1; \dotsc, k_p, c_p)}_{n+p,h}(x_1, \dotsc, x_n)
\end{equation}
Plugging it into \eqref{nPointEquation} leads to
\begin{multline} \label{TopologicalLoopEquation}
\sum_{\substack{(I_1, I_2)\in \mathcal{I}_2(n)\\ h=0, \dotsc, g}} W_{|I_1|+1, h}(x_1, x_{I_1}) W_{|I_2|+1, g-h}(x_1,x_{I_2}) + W_{n+1, g-1}(x_1, x_1, \dotsc, x_n) \\
+ \sum_{j=2}^n \mathbbm{1}(x_1,x_j) \frac{\partial}{\partial x_j} \frac{W_{n-1, g}(x_2, \dotsc, x_n) - W_{n-1, g}(x_2, \dotsc, x_{j-1}, x_1, x_j, \dotsc, x_n)}{x_j-x_1}\\
+ \sum_{\substack{\vlam\\ R\in\mathcal{P}(L(\vlam))}} \sum_{(I_1, \dotsc, I_{\ell(R)}) \in \mathcal{I}_{\ell(R)}(n)} \sum_{\substack{h, h_1, \dotsc, h_{\ell(R)}\geq 0\\ \ell(\vlam) - \ell(R) + h + \sum_{\alpha=1}^{\ell(R)} h_{\alpha}= g}} \sum_{c=1}^d \mathbbm{1}(x_1,c) \sum_{j=1}^{\ell(\lambda^{(c)})} S^{(h)}(\vlam) \prod_{\substack{\alpha\\ R_{\alpha}\neq R(c,j)}} W_{|R_\alpha|+ |I_\alpha|, h_\alpha}^{(\lambda_{R_\alpha}, c_{R_\alpha})}(x_{I_\alpha})\\
\times \lambda^{(c)}_j \Bigl(x_1^{\lambda^{(c)}_j-1} W_{|R(c,j)|+|I(c,j)|, h(c,j)}^{(\lambda_{R'(c,j)}, c_{R'(c,j)})}(x_1,x_{I(c,j)}) - \sum_{q=0}^{\lambda^{(c)}_j-2} x_1^{\lambda^{(c)}_j-2-q} W_{|R(c,j)|+|I(c,j)|, h(c,j)}^{(\lambda_{R'(c,j)}, c_{R'(c,j)}; q,c)}(x_{I(c,j)}\Bigr)  = 0
\end{multline}
where $h(c,j)$ is the $h_{\alpha^*}$ where $\alpha^*$ is such that $R_{\alpha^*} = R(c,j)$.

\subsection{Large \texorpdfstring{$N$} limit}

Restricting \eqref{TopologicalLoopEquation} to $g=0$ gives the constraint $h=h_1=\dotsb=h_{\ell(R)}=0$, which in turn implies $\ell(R) = \ell(\vlam)$. This reduces the sum over partitions of $L(\vlam)$ to a single one for each $\vlam$, i.e. the partition into singletons, $R=\{\{c',i\}\}$, for $c'=1, \dotsc, d$ and $i=1, \dotsc, \ell(\lambda^{(c')})$. For $(I_1, \dotsc, I_{\ell(\vlam)})\in \mathcal{I}_{\ell(\vlam)}(n)$, we denote $I(c',i) = I_\alpha$ when $R_{\alpha} = \{c',i\}$. This gives, for a generic potential,
\begin{multline} \label{PlanarLoopEquation}
\sum_{(I_1, I_2)\in \mathcal{I}_2(n)} W_{|I_1|+1, 0}(x_1, x_{I_1}) W_{|I_2|+1, 0}(x_1,x_{I_2}) \\
+ \sum_{j=2}^n \mathbbm{1}(x_1,x_j) \frac{\partial}{\partial x_j} \frac{W_{n-1, 0}(x_2, \dotsc, x_n) - W_{n-1, 0}(x_2, \dotsc, x_{j-1}, x_1, x_j, \dotsc, x_n)}{x_j-x_1}\\
+ \sum_{\substack{\vlam\\ (I_1, \dotsc, I_{\ell(\vlam)})\in \mathcal{I}_{\ell(\vlam)}(n)}} \sum_{c=1}^d \mathbbm{1}(x_1,c) \sum_{j=1}^{\ell(\lambda^{(c)})} S^{(0)}(\vlam) \prod_{(c',i)\neq (c,j)} W_{|I(c',i)|+1, 0}^{(\lambda^{(c')}_i,c')}(x_{I(c',i)})\\
\times \lambda^{(c)}_j \Bigl(x_1^{\lambda^{(c)}_j-1} W_{|I(c,j)|+1, 0}(x_1,x_{I(c,j)}) - \sum_{q=0}^{\lambda^{(c)}_j-2} x_1^{\lambda^{(c)}_j-2-q} W_{|I(c,j)|+1, 0}^{(q,c)}(x_{I(c,j)}\Bigr)  = 0.
\end{multline}

In the case of $Z_{\text{Fluct}}(N, \mathcal{P})$, Equation \eqref{ShiftedPartitionFunction}, which we are interested in, the leading order coefficients $S^{(0)}(\vlam)$ actually have an extra $N$-dependence, satisfies \eqref{GaussianReduction}. It implies that the partitions $\vlam$ appearing in \eqref{PlanarLoopEquation} must be of size 2, $|\vlam|=2$. These partitions are of the form $\vlam=(\lambda^{(1)}, \dotsc, \lambda^{(d)})$ with
\begin{itemize}
\item either $\lambda^{(c)} = (2)$ for some $c$. We write the corresponding part of the action $\frac{1}{2}\sum_{c=1}^d a_c \tr_{E_c}M_c^2$.
\item or $\lambda^{(c)}=(1,1)$ for some $c$. We write the corresponding part of the action $\frac{1}{2}\sum_{c=1}^d b_{cc} \bigl(\tr_{E_c}M_c\bigr)^2$.
\item or $\lambda^{(c)}=(1)$ and $\lambda^{(c')}=(1)$ for some $c\neq c'$.

We write the corresponding part of the action $\frac{1}{2}\sum_{c\neq c'}^d b_{cc'} \tr_{E_c}M_c\ \tr_{E_{c'}}M_{c'}$, with $b_{cc'}=b_{c'c}$.
\end{itemize}
This means that at large $N$, the correlation functions behave as if the action simply was
\begin{equation} \label{GaussianCoefficients}
-\frac{1}{2}N\sum_c \tr_{E_c}M_c^2 - \sum_{|\vlam|= 2} N^{2-\ell(\vlam)} s^{(0)}_{\mathcal{P}}(\vlam)\ I_{\vlam}(\{M_c\}) = \frac{1}{2}N \sum_c a_c \tr_{E_c}M_c^2 + \frac{1}{2}\sum_{c,c'} b_{c, c'} \tr_{E_c}M_c\, \tr_{E_{c'}} M_{c'}.
\end{equation}
Then \eqref{PlanarLoopEquation} becomes
\begin{multline} \label{PlanarGaussianLoopEquation}
\sum_{(I_1, I_2)\in \mathcal{I}_2(n)} W_{|I_1|+1, 0}(x_1, x_{I_1}) W_{|I_2|+1, 0}(x_1,x_{I_2}) \\
+ \sum_{j=2}^n \mathbbm{1}(x_1,x_j) \frac{\partial}{\partial x_j} \frac{W_{n-1, 0}(x_2, \dotsc, x_n) - W_{n-1, 0}(x_2, \dotsc, x_{j-1}, x_1, x_j, \dotsc, x_n)}{x_j-x_1}\\
- \sum_{c=1}^d \mathbbm{1}(x_1,c) \Bigl(a_c x_1 W_n(x_1, \dotsc, x_n) + \sum_{(I_1, I_2)\in \mathcal{I}_2(n)} W_{|I_1|+1,0}(x_1, x_{I_1}) \sum_{c'=1}^d b_{cc'} W_{|I_2|+1,0}^{(1,c')}(x_{I_2})\Bigr) = 0.
\end{multline}

\subsubsection{Disc function}

The disc function of color $c$ is $W_{1,0}(x,c)=\lim_{N\to\infty} \frac{1}{N} W_1(x,c)$, i.e. the generating function of planar stuffed maps with a single boundary of arbitrary perimeter. The global disc function is $W_{0,1}(x) = \sum_{c=1}^d \mathbbm{1}(x,c) W_{0,1}(x,c)$. 

Equation \eqref{PlanarGaussianLoopEquation} can be directly applied. Before doing so, however, let us briefly mention the case of a generic potential, by setting $n=1$ in \eqref{PlanarLoopEquation}. For a fixed color $c$,
\begin{equation} \label{GenericDiscEquation}
W_{0,1}(x,c)^2 + \sum_{\vlam} \sum_{j=1}^{\ell(\lambda^{(c)})} s^{(0)}(\vlam) \biggl(\prod_{(c',i)\neq (c,j)} W_{1,0}^{(\lambda^{(c')}_i,c')}\biggr) \lambda^{(c)}_j \Bigl(x^{\lambda^{(c)}_j-1} W_{1,0}(x,c) - \sum_{q=0}^{\lambda^{(c)}_j-2} x^{\lambda^{(c)}_j-2-q} W_{1,0}^{(q,c)}\Bigr) = 0
\end{equation}
This is thus a set of $d$ equations on $d$ functions $W_{0,1}(x,c)$ with a single catalytic variable, and some explicit dependence on coefficients of the unknown series. This generalizes the classical equation of the 1-matrix, multi-trace model \cite{BlobbedTR}. The analytic properties of its disc function, described in \cite{BorotBouttierGuitter}, derive from an extension of \cite{BM-Jehanne}, which applies to stuffed maps with bounded face degrees, to stuffed maps with unbounded face degrees. Here, we would require a further extension, to a system of equations.

Instead of pursuing the generic route, we focus on the specific model $Z_{\text{Fluct}}(N, \mathcal{P})$. Setting $n=1$ in \eqref{PlanarGaussianLoopEquation}, one finds
\begin{equation}
W_{1,0}(x,c)^2 - a_c \bigl(xW_{1,0}(x,c) - 1\bigr) = 0.
\end{equation}
where we have used $W_{1,0}^{(1,c')} = 0$ for all $c'$, since $W_{1,0}^{(1,c')} = \lim_{N\to\infty} \frac{1}{N} \langle \tr_{E_{c'}} M_{c'}\rangle$ and the model is invariant under $\{M_c\} \to \{-M_c\}$ at large $N$. The disc function of color $c$ is thus
\begin{equation}
W_{1,0}(x,c) = \frac{a_c}{2} \biggl(x - \sqrt{x^2 - \frac{4}{a_c}}\biggr)
\end{equation}
as in \cite{QuarticTR}. It has a cut along $\Gamma_c = [-\frac{2}{\sqrt{a_c}}, \frac{2}{\sqrt{a_c}}]$ (if $a_c>0$), which is said to be the cut of color $c$. The global disc function $W_{1,0}(x) = \sum_{c=1}^d W_1^{(0)}(x,c) \mathbbm{1}_{\C_c}(x)$ thus has $d$ cuts, along $\bigcup_{c=1}^d \Gamma_c$.

\subsubsection{Cylinder function}

The cylinder function is the leading order of the two-point function. The local expression is
\begin{equation}
W_{2,0}(x_1, c_1; x_2, c_2) = \lim_{N\to\infty} W_{2}(x_1,c_1;x_2,c_2) = \lim_{N\to\infty} \langle \tr_{E_{c_1}}\frac{1}{x_1-M_{c_1}} \tr_{V_{c_2}}\frac{1}{x_2-M_{c_2}}\rangle_{\text{conn}}
\end{equation}

By setting $n=2$ in \eqref{PlanarGaussianLoopEquation}, one finds
\begin{multline}
\bigl(2 W_{1,0}(x_1, c_1) - a_{c_1} x_1\bigr) W_{2,0}(x_1, c_1; x_2, c_2) + \delta_{c_1,c_2} \frac{\partial}{\partial x_2} \frac{W_{1,0}(x_1, c_1) - W_{1,0}(x_2, c_2)}{x_1-x_2} \\
 - W_{1,0}(x_1, c_1) \sum_{c'=1}^d b_{c_1c'} W_{2,0}^{(1,c')}(x_2, c_2) =0
\end{multline}
where a term $- W^{(0)}_{2}(x_1, c_1; x_2, c_2) \sum_{c'=1}^d b_{c_1c'} W_{1,0}^{(1,c')}$ has been removed because $W_{1,0}^{(1,c')}=0$ in this model.

It generalizes the equations found for $W_{2}^{(0)}(x_1, c_1; x_2, c_2)$ in \cite{QuarticTR}. The method used to solve them still works here. It processes by first finding the values of $W^{(1,c')}_{2,0}(x_2, c_2)$. To do so, we extract the coefficient of the equations at order $1/x_1$. It gives
\begin{equation}
- a_{c_1} W_{2,0}^{(1,c_1)}(x_2, c_2) - \sum_{c'=1}^d b_{c_1,c'} W_{2,0}^{(1,c')}(x_2, c_2) = \delta_{c_1,c_2} \frac{\partial W_{1,0}(x_2,c_2)}{\partial x_2}
\end{equation}
which can be given a matrix form. Introduce the following $d\times d$ matrices $A = \operatorname{diag}(a_1, \dotsc, a_d)$, $B = (b_{c,c'})_{1\leq c,c'\leq d}$ and 
\begin{equation}
W^{(1)}_{2,0}(x) = \bigl(W^{(1,c)}_{2,0}(x,c')\bigr)_{1\leq c,c'\leq d}\qquad \partial W_{1,0}(x_2) = \operatorname{diag}\Bigl(\frac{\partial W^{(0)}_{1}(x_2)}{\partial x_2}, \dotsc, \frac{\partial W^{(0)}_{d}(x_2)}{\partial x_2}\Bigr)
\end{equation}
It comes
\begin{equation}
W^{(1)}_{2,0}(x) = - (A+B)^{-1} \partial W_{1,0}(x)
\end{equation}
provided $A+B$ is invertible.

Denoting $\sigma(x,c) = \sqrt{x^2-4/a_c}$, we have
\begin{equation}
\frac{\partial W_{1,0}(x, c)}{\partial x} = -\frac{W_{1,0}(x, c)}{\sigma(x,c)}
\end{equation}
and it comes
\begin{equation} \label{W20}
W_{2,0}(x_1, c_1; x_2, c_2) = \delta_{c_1c_2} \frac{x_1 x_2 - \sigma(x_1,c_1) \sigma(x_2,c_2) - 4/a_{c_1}}{2(x_1 - x_2)^2 \sigma(x_1,c_1) \sigma(x_2,c_2)} - \frac{1}{a_{c_1}} \Bigl(B \frac{1}{A+B}\Bigr)_{c_1c_2} \frac{W_{1,0}(x_1, c_1) W_{1,0}(x_2, c_2)}{\sigma(x_1,c_1) \sigma(x_2,c_2)}
\end{equation}
The first part of this formula is the cylinder function for the GUE, while the other term is due to the multi-trace interaction. Notice that the latter is not manifestly invariant under the exchange symmetry $(x_1, c_1) \leftrightarrow (x_2, c_2)$.

The cylinder function is thus basically the same as in the quartic melonic model \cite{QuarticTR}. The difference is that in that case, the matrices $A, B$ were specific functions of the coupling constants $t_1, \dotsc, t_d$, while they are here, in general, functions of all the coupling constants $t_i$, $i\in I$.

\subsection{Blobbed topological recursion}

The remaining is exactly similar to \cite{QuarticTR}, with the replacement $a_c = 1-\alpha_c^2$, which already followed the theorems of \cite{BlobbedTR, BlobbedTR2}.

\subsubsection{Spectral curve}

Denote the Riemann sphere $\hat{\C}$, and $\hat{\C}_c$ its copy of color $c\in\{1, \dotsc, d\}$. For each color, define
\begin{equation}
f_c(x,y) = y^2 - a_c xy + a_c
\end{equation}
The Gaussian spectral curve is defined as $\mathcal{C}_{\text{Gaussian}}\subset \bigcup_{c=1}^d\hat{\C}_c^2$ by
\begin{equation}
f(x,y) = \sum_{c=1}^d \mathbbm{1}(x,c) \mathbbm{1}(y,c) f_c(x,y) = 0
\end{equation}

Recall that $\Gamma_c = [-\frac{2}{\sqrt{a_c}}, \frac{2}{\sqrt{a_c}}]$ (for $a_c>0$) and denote $\Gamma = \bigcup_{c=1}^d \Gamma_c$. It can be checked as in \cite{BlobbedTR} that our correlation functions $W_{n,g}(x_1, \dotsc, x_n)$ are only singular along $\Gamma$ (with respect to each variable), except for $(n,g)=(1,0),(2,0)$. We therefore introduce a Zhukovski parametrization 
\begin{equation}
x(z) = \sum_{c=1}^d \mathbbm{1}(x,c) \mathbbm{1}(z,c) \frac{1}{\sqrt{a_c}}(z+z^{-1})
\end{equation}
for $|z|\geq 1$ in each color. As shown in \cite{BlobbedTR}, the correlation functions for $(n,g)\neq(1,0),(2,0)$ are holomorphic for $|z|\geq 1$ except at $z=\pm 1$. Moreover, they can be analytically continued to the interior of a neighborhood of the unit circle, except at $z=\pm 1$. In our case (see below), this analytic continuation can be performed for all $0<|z|<1$. The correlation functions can thus be turned into differential forms
\begin{equation}
\omega_{n,g}(z_1, \dotsc, z_n) = W_{n,g}(x(z_1), \dotsc, x(z_n)) dx(z_1) \dotsm dx(z_n) + \delta_{(n,g),(2,0)} \sum_{c=1}^d \mathbbm{1}(z_1,c) \mathbbm{1}(z_2,c) \frac{dx(z_2) dx(z_2)}{(x(z_1)-x(z_2))^2}.
\end{equation}
For $(n,g)\neq (1,0),(2,0)$, they are holomorphic on $\mathcal{C}_{\text{Gaussian}}^n$ except at $z_i = 0, \pm 1$.

In the framework of the blobbed topological recursion, the singularities at 0 and at $\pm 1$ play different roles. This is because the disc function has (simple) zeroes at $z=\pm 1$,
\begin{equation}
W_{1,0}(x(z)) = \sum_{c=1}^d \mathbbm{1}(z,c) \frac{\sqrt{a_c}}{z} \quad \Rightarrow \quad \omega_{1,0}(z) = \sum_{c=1}^d \mathbbm{1}(z,c) \frac{\sqrt{a_c}}{z} (1-z^{-2}) dz
\end{equation}
because they are the zeroes of $dx(z)$.

As for the cylinder function, it becomes
\begin{equation} \label{omega20}
\omega_{2,0}(z_1,z_2) = \frac{dz_1\,dz_2}{(z_1-z_2)^2} \sum_{c=1}^d \mathbbm{1}(z_1,c) \mathbbm{1}(z_2,c) - \frac{dz_1 dz_2}{z_1^2 z_2^2}\sum_{c_1, c_2=1}^d \frac{\mathbbm{1}(z_1,c_1) \mathbbm{1}(z_2,c_2)}{\sqrt{a_{c_1} a_{c_2}}} \frac{1}{a_{c_1}} \Bigl(B \frac{1}{A+B}\Bigr)_{c_1c_2}
\end{equation}
which is singular along the diagonal $z_1=z_2$ as expected, but also has poles at $z_1=0$ and $z_2=0$ on each color.

The points $z=\pm 1$ in each color are called the ramification points and
\begin{equation}
\mathcal{R} = \{z=\pm 1 \in\C_c, c=1, \dotsc, d\}
\end{equation}
The spectral curve is also supplemented with the canonical involution $\iota(z)=1/z$ which preserves the ramification points.

\subsubsection{Topological recursion formula}

Denote $G(z,z_1) = \int^{z} \omega_{2,0}(\cdot, z_1)$, and $\Delta \varphi = \varphi - \iota^*\varphi$ for a differential form $\varphi$. The kernel of the topological recursion is $K(z, z_1) = \frac{\Delta G(z,z_1)}{2\Delta \omega_{1,0}(z)}$. We further define the polar and holomorphic part of $\omega_{n,g}(z_1, \dotsc, z_n)$ as follows (terminology justified below)
\begin{equation}
P\omega_{n,g}(z_1, \dotsc, z_n) = \sum_{z\in\mathcal{R}} \operatorname{res}_z G(z,z_1) \omega_{n,g}(z, \dotsc, z_n),\qquad H\omega_{n,g}(z_1, \dotsc, z_n) = \omega_{n,g}(z_1, \dotsc, z_n) - P\omega_{n,g}(z_1, \dotsc, z_n)
\end{equation}

\begin{theorem} \label{thm:BlobbedTR}
Assume that $A+B$ is invertible, so that $\omega_{2,0}$ is given by \eqref{omega20}. For all $(n,g)\neq (1,0),(2,0)$, the holomorphic part is holomorphic while the polar part has poles on $\mathcal{R}$. They are given by
\begin{equation}
P\omega_{n,g}(z_1, \dotsc, z_n) = \sum_{z\in\mathcal{R}} \operatorname{res}_z K(z,z_1)\Bigl( \omega_{n+1,g-1}(z,\iota(z),z_2, \dotsc, z_n) + \sum_{\substack{(I_1, I_2)\in\mathcal{I}_2(n)\\ h=0, \dotsc, g\\ (|I_1|,h) \neq (1,0)\\ (|I_2|,h)\neq(1,g)}} \omega_{|I_1|+1,h}(z, z_{I_1}) \omega_{|I_2|+1}(\iota(z), z_{I_2})\Bigr)
\end{equation}
and
\begin{equation}
H \omega_{n,g}(z_1, \dotsc, z_n) = \frac{1}{2i\pi} \oint_{z\in\bigcup_{c=1}^d \mathbbm{U}_c} \omega_{2,0}(z_1,z) \nu_{n,g}(z, z_2, \dotsc, z_n) 
\end{equation}
where $\mathbbm{U}_c$ is the copy of the unit circle of color $c$, and 
\begin{equation}
\nu_{n,g}(z, z_2, \dotsc, z_n) = V_{n,g}(x(z), x(z_2), \dotsc, x(z_n))\ dx(z_2) \dotsm dx(z_n)
\end{equation}
for
\begin{multline}
V_{n,g}(x, x_2, \dotsc, x_n) = \sum_{c=1}^d \mathbbm{1}(x,c) \\
\sum_{\vlam} \sum_{j=1}^{\ell(\lambda^{(c)})} \sum_{\substack{R\in\mathcal{P}(L(\vlam)) \\ R'(c,j)=\emptyset}} \sum_{(I_1, \dotsc, I_{\ell(R)}) \in \mathcal{I}_{\ell(R)}(n)} \sum_{\substack{h, h_1, \dotsc, h_{\ell(R)}\geq 0\\ \ell(\vlam) - \ell(R) + h + \sum_{\alpha=1}^{\ell(R)} h_{\alpha}= g}} S^{(h)}(\vlam) x^{\lambda^{(c)}_j} \prod_{R_\alpha\neq \{(c,j)\}} W_{|R_\alpha|+|I_\alpha|, h_\alpha}^{(\lambda_{R_\alpha}, c_{R_\alpha})} (x_{I_\alpha})
\end{multline}
\end{theorem}

\begin{proof}
The loop equations have the same form as in \cite{QuarticTR} and all the arguments, which were borrowed from \cite{BlobbedTR, BlobbedTR2} apply.
\end{proof}

\section*{Conclusion}

We have shown that, as long as there are quartic melonic interactions, one can find in arbitrary tensor models a set of correlation functions which satisfy the blobbed topological recursion in a universal way. The spectral curve is a disjoint union of Gaussian spectral curves, with an additional holomorphic part to the cylinder function which always has the same form.

Those results rely on the conditions 1 and 2 presented in the introduction and detailed throughout the article. In particular, the specifics of the model, i.e. the choice of interactions, do not matter as long as the effective action obtained after the formal integration of all the matrices except $Y_1, \dotsc, Y_d$ has a well-defined $1/N$ expansion as we have described in Theorem \ref{thm:PartialIntegration}. This is why our formulas all have the same structure as in the case of the quartic melonic model in \cite{QuarticTR}.

We have also provided theorems \ref{thm:Expectations}, \ref{thm:IF} and \ref{thm:BubbleExpectation} to relate the expectations of the $U(N)^d$-invariant observables on the tensor and matrix sides.

There are still many interesting questions about the topological recursion for tensor models. Are there other sets of correlation functions satisfying the topological recursion? Is it always a blobbed recursion? Can condition 2 from the introduction be removed? Is it possible to proceed without going to matrix models and derive topological recursions directly in the tensor formulation? There have been some efforts to use directly the Schwinger-Dyson equations of tensor models, in \cite{SDLargeN} and \cite{DoubleScaling}, thus extracting the double scaling limit of tensor models with melonic interactions for instance, but this is still far from the topological recursion. We hope some of those questions can be tackled in the near future.

\end{document}